\documentclass[sigconf,nonacm,]{acmart}
\usepackage{amsmath}
\usepackage{graphicx}
\usepackage{subcaption}

\usepackage{caption}
\usepackage{hyperref}
\usepackage{enumitem}

\usepackage{bm}

\usepackage{amsfonts}

\usepackage[linesnumbered,vlined,ruled]{algorithm2e}
\usepackage{cases}
\usepackage{bm}
\usepackage{url}

\usepackage{amsthm}
\usepackage{booktabs}
\usepackage{xcolor}
\usepackage{epstopdf}
\usepackage{bbm} 
\usepackage{mathrsfs}
\usepackage{multicol}
\usepackage{etoolbox}
\usepackage{mathtools}
\usepackage{tabularx} 
\usepackage{multirow}
\usepackage{diagbox}
\usepackage{float}
\usepackage{pifont}
\usepackage{pbox}

\usepackage{xspace}
\usepackage[normalem]{ulem}
\usepackage[export, demo]{adjustbox}
\usepackage{array,multirow,graphicx}

\useunder{\uline}{\ul}{}

\usepackage{threeparttable}

\theoremstyle{definition}

\SetKw{KwInput}{Input}
\SetKw{KwOutput}{Output}

\theoremstyle{remark}

\theoremstyle{plain}
\newtheorem{theorem}{Theorem}

\newcommand{\our}[0]{\texttt{ESD}\xspace}
\newcommand{\ouryi}[0]{\texttt{ESD($\alpha = 1$)}\xspace}
\newcommand{\ourwu}[0]{\texttt{ESD($\alpha = 0.5$)}\xspace}
\newcommand{\ourerwu}[0]{\texttt{ESD($\alpha = 0.25$)}\xspace}
\newcommand{\ouryierwu}[0]{\texttt{ESD($\alpha = 0.125$)}\xspace}
\newcommand{\ourling}[0]{\texttt{ESD($\alpha = 0$)}\xspace}

\newcommand{\ourlru}[0]{\texttt{Emark}\xspace}

\newcommand{\opt}[0]{\texttt{Opt}\xspace}
\newcommand{\heu}[0]{\texttt{Heu}\xspace}
\newcommand{\ourmix}[0]{\texttt{HybridDis}\xspace}

\newcommand{\target}[0]{\texttt{target}\xspace}

\newcommand{\stitle}[1]{\vspace{.5ex}\noindent{\bf #1}}

\newcommand{\tann}[1]{\textcolor{black}{#1}}

\newcommand{\ie}{\emph{i.e.}}
\newcommand{\eg}{\emph{e.g.}}
\newcommand{\etc}{\emph{etc.}}
\newcommand{\etal}{\emph{et al.}}

\begin{document}

\title{Embedding Samples Dispatching for Recommendation Model Training in Edge Environments}

\author{Guopeng Li}
\affiliation{
  \institution{USTC}
  \country{Hefei, China}
}
\author{Haisheng Tan}
\authornote{Corresponding author.}
\affiliation{
  \institution{USTC}
  \country{Hefei, China}
}
\author{Chi Zhang}
\authornotemark[1]
\affiliation{
  \institution{Hefei University of Technology}
  \country{Hefei, China}
}
\author{Hongqiu Ni}
\affiliation{
  \institution{USTC}
  \country{Hefei, China}
}

\author{Zilong Wang}
\affiliation{
  \institution{USTC}
  \country{Hefei, China}
}
\author{Xinyue Zhang}
\affiliation{
  \institution{USTC}
  \country{Hefei, China}
}
\author{Yang Xu}
\affiliation{
  \institution{USTC}
  \country{Hefei, China}
}
\author{Han Tian}
\affiliation{
  \institution{USTC}
  \country{Hefei, China}
}


\begin{abstract}
Training deep learning recommendation models (DLRMs) on edge workers brings  several benefits, particularly in terms of data privacy protection, low latency and personalization. However, due to the huge size of embedding tables, typical DLRM training frameworks adopt one or more parameter servers to maintain global embedding tables, while leveraging the edge workers cache part of them. This incurs significant transmission cost for embedding transmissions between workers and parameter servers, which can dominate the training cycle. In this paper, we investigate how to dispatch input embedding samples to appropriate edge workers to minimize the total embedding transmission cost when facing edge-specific challenges such as heterogeneous networks and limited resources. We develop \our, a novel mechanism that optimizes the dispatch of input embedding samples to edge workers based on expected embedding transmission cost. We propose \ourmix as the dispatch decision method within \our, which combines a resource-intensive optimal algorithm and a heuristic algorithm to balance decision quality and resource consumption. We implement a prototype of \our and compare it with state-of-the-art mechanisms on real-world workloads. Extensive experimental results show that \our reduces the embedding transmission cost by up to $36.76\%$ and achieves up to $1.74\times$ speedup in end-to-end DLRM training.
\end{abstract}


\maketitle
\thispagestyle{plain} 
\pagestyle{plain}    

\section{Introduction}\label{sec:intro}

Recommender systems have become essential in daily life, with recommendation models trained on large-scale data to learn user preferences and product characteristics, providing personalized recommendations in e-commerce~\cite{gu2021self, wag2018billio, smith2017two}, online entertainment \cite{gomez2015netflix,covington2016deep,schedl2017new}, and social media~\cite{sharma2016graphjet,boeker2022empirical,ying2018graph} to enhance user experience and drive business growth~\cite{chen2023openembedding,wang2024atrec}. In recent years, with the rapid development of mobile and ubiquitous computing, integrating recommender systems into mobile and  ubiquitous applications to provide personalized services has become a growing trend, such as energy management recommendations in the Internet of Energy~\cite{sayed2021intelligent,himeur2021survey}, health advice in health monitoring~\cite{pourpanah2024exploring}, and destination  suggestions in tourism exhibitions~\cite{su2019edge}. However, conventional recommender systems are typically deployed on cloud servers, which face data privacy concerns and network latency challenges~\cite{yin2024device,cai2024distributed,long2024diffusion,yuan2024hetefedrec}. To deliver recommendation services while ensuring privacy protection and low latency for mobile and  ubiquitous applications, training and deploying recommendation models on edge workers is a natural alternative~\cite{dac2023,gong2020edgerec,himeur2022latest}.

 \underline{D}eep \underline{L}earning \underline{R}ecommendation \underline{M}odels (DLRMs), such as 
WDL~\cite{cheng2016wide} and DFM~\cite{guo2017deepfm}, have emerged as highly effective techniques for constructing recommender systems~\cite{jiang2023mixrec}.  A typical  DLRM follows the \textit{embedding layer}  and \textit{ multi-layer perceptron (MLP)}   paradigm. The embedding layer is practically the embedding tables,  transforming the sparse input (\eg, user ID and video  ID) into  dense vector embeddings~\cite{guo2021scalefreectr,zhao2023embedding,zhang2024cafe}. The embedding tables  consume a tremendous memory footprint, with each embedding requiring a few KB, and the total memory reaching tens of GBs to TBs. In contrast, the MLP layers, typically ranging from several to hundreds of MBs, are much smaller, and make up the dense model that learns high-order feature interactions and generates the final prediction \cite{Zhou2018deep,zhag2022picasso}. 

When training DLRM on edge workers for mobile and  ubiquitous recommendation services,  the parameter server (PS) architecture is commonly employed to address the issue of huge embedding tables and the limited memory capacity of  edge workers. In the  PS architecture,  one or more memory-optimized servers (\ie, the parameter servers) maintain the globally embedding tables, while each edge worker caches a portion of embedding tables in its local memory, \ie, \textit{embedding cache}. During DLRM training,  workers pull embedding values from the parameter server and push embedding gradients back through the Ethernet, called \textit{embedding transmission}. In production workloads, a single training input sample may involve up to thousands of embeddings ~\cite{acun2021understanding}. Training DLRM  on  edge workers suffers from embedding access bottlenecks as each training iteration requires substantial embedding transmission between edge workers and the parameter server~\cite{song2023ugache}. For instance, training WDL~\cite{cheng2016wide}  with the Criteo Kaggle dataset~\cite{criteodata}  shows that  up to $90\%$ time could be spent on embedding transmission, dominating the training cycle.

In this paper, we aim to reduce the network bandwidth usage, denoted as the \emph{ total embedding transmission cost}, in resource-limited edge environments while simultaneously accelerating DLRM training. Managing the input embedding samples loaded from the data loader for each worker is an effective way to reduce the embedding transmission cost~\cite{zeng2024accelerating,agarwal2023bagpipe,kwon2022training}. Specifically, with input samples prefetching in today's data loaders~\cite{paszke2019pytorch,ma2023fec}, during the current training iteration, input samples for the next iteration can be proactively dispatched to the appropriate worker to minimize the embedding transmission cost.  Intuitively, we may dispatch input samples to the workers that maximize the overall cache hit ratio. However, unlike traditional file caching, embeddings stored on workers during DLRM training are not static, which are continuously updated through synchronization with the parameter server, making the dispatch problem non-trivial. Moreover, the heterogeneous networks and limited resources in edge environments  also bring  hardness. We summarize the problem challenges as follows.

\stitle{Cost composition.} When training DLRM with the PS architecture, it is crucial that each worker caches the latest versions of embeddings corresponding to the current input samples. If the latest embeddings are not available in the cache, they must be pulled from the parameter server (\textit{miss pull})\footnote{This paper employs Bulk Synchronous Parallel (BSP)~\cite{ghadimi2016mini} training for DLRM without sacrificing model accuracy, necessitating the use of the latest embeddings.}. For each training iteration, embedding gradients are on-demand synchronized between the workers and the parameter server (\textit{update push}). Additionally, when the embedding cache reaches its capacity, certain embeddings must be evicted to make space for new ones, and their gradients have to be transmitted to the parameter server (\textit{evict push}). Consequently, the embedding transmission cost arises from the \textit{miss pull}, \textit{update push}, and \textit{evict push} operations. Therefore, a mechanism is desired to minimize the total embedding transmission cost rather than merely optimizing the hit ratio.

\stitle{Heterogeneous networks.} When training DLRM in edge environments, there might be significantly lower-bandwidth connections between edge workers and the parameter server compared to the high-bandwidth networks used in cloud data centers (\eg, 100 Gbps Ethernet or InfiniBand). Typically, edge workers connect to the parameter server via heterogeneous networks. For instance, some workers may use 0.5 Gbps connections while others rely on 5 Gbps, leading to varying transmission costs for the embeddings of the same size~\cite{park2020mobile,li2020optimized,um2024metis}.   
Therefore, beyond accounting for the number of miss pull, update push, and evict push operations, it is also necessary to consider the heterogneous network bandwidths between workers and the parameter server.

\stitle{Limited resources.} In production environments, each iteration typically involves an input batch of thousands of embedding samples, and dispatching these samples to different workers introduces decision-making overhead. In edge environments, where multiple tasks compete for limited resources~\cite{ling2021rt,kong2023accumo}, this overhead must be carefully managed. Furthermore, in the online training scenario considered in this paper, which aims to promptly and accurately capture user interest drift and emerging trends~\cite{zhao2023recd,zhag2022picasso}, dispatch decisions must be completed within each training iteration to avoid performance degradation. Therefore, when designing the dispatch mechanism, it is crucial to account for both the time and resource consumption of dispatch decision-making in light of the challenges posed by edge environments and online training requirements.

Prior works focus on mitigating embedding transmission cost when training DLRM in cloud environments and overlook the consideration for edge-specific challenges~\cite{song2023ugache,miao2021het,adnan2021accelerating,zeng2024accelerating,ma2023fec}. To address the above practical challenges, we propose \our, a mechanism for dispatching embedding samples to edge workers to reduce the total embedding transmission cost and accelerate DLRM training. \our assigns an expected transmission cost for training each sample on different edge workers and dispatches  samples to  appropriate workers to reduce the total embedding transmission cost. To avoid decision time exceeding training time while balancing decision quality and resource consumption, we propose \ourmix as the dispatch decision method in \our,  which combines a resource-intensive optimal algorithm with a resource-efficient heuristic algorithm. Our technical contributions are summarized as follows:

\begin{itemize}[leftmargin=*]
    \item In this paper, we formulate the problem of dispatching  embedding samples to minimize the total embedding transmission cost during DLRM training in edge environments, \tann{taking into account cost composition, heterogeneous networks and limited resources.} To the best of our knowledge, this is the first work specifically focused on minimizing embedding transmission cost in edge environments (Sec.~\ref{sec:formulation}).
    \item We introduce \our, an embedding sample dispatch mechanism that reduces the total embedding transmission cost by calculating the expected costs for dispatching each embedding sample to different workers.  We propose \ourmix, a hybrid dispatch decision method within \our that balances decision quality and resource consumption to address the limited resources challenge (Sec.~\ref{sec:design} ).
    \item We implement \our using C++ (including CUDA) and Python to evaluate it through extensive experiments on a testbed with $8$ edge workers under real-world workloads. Evaluation results show that compared with LAIA~\cite{zeng2024accelerating}, \our reduces embedding transmissions by up to $36.76\%$ and delivers up to $1.74\times$ performance improvement for end-to-end DLRM training (Sec.~\ref{sec:imple} and Sec.~\ref{sec:eva}).
\end{itemize}

\section{Background}\label{sec:back}

\subsection{Recommender System for Mobile and  Ubiquitous Applications  in Edge Environments}

Recommender systems are now essential in e-commerce, social media, and media applications, offering personalized recommendations that enhance user experience and drive business growth~\cite{wang2024autosr,li2022towards}. Traditional recommender systems are cloud-based, where the cloud server trains the model using all user data and pushes recommendation results to users' devices upon request. Though this paradigm  benefits from the infinite computing power, it faces  challenges, including high resource  consumption~\cite{wang2020next}, dependence on network access for timeliness~\cite{long2023decentralized}, and user privacy concerns~\cite{muhammad2020fedfast,sun2024survey,zhang2024gpfedrec,ding2024fedloca}, which undermine the sustainability and trustworthiness of cloud-based systems~\cite{belal2022pepper}. Furthermore,  cloud-based systems often fail to meet the time-sensitive and computational demands of mobile and  ubiquitous applications.  With the rise of edge computing, there is a  trend towards training and deploying recommendation models in edge environments~\cite{dac2023,gong2020edgerec,himeur2022latest}.  This paradigm provides recommendation services while ensuring privacy protection and low latency for mobile and  ubiquitous applications such as the Internet of Energy~\cite{sayed2021intelligent,himeur2021survey,xia2023reca}, health monitoring~\cite{pourpanah2024exploring,gao2020monitoring}, and tourism exhibitions~\cite{su2019edge}.  For instance, in the Internet of Energy, recommender systems can suggest energy management solutions and assist in resource allocation. In health monitoring, they provide personalized health advice based on medical history and real-time health information. 

For data privacy, this paper achieves privacy by training DLRM on edge workers within a local area network (LAN) isolated from the external internet. This approach keeps private data local, preventing leakage while enabling low-latency recommendation services. Moreover, this paper focuses on online DLRM training. At the system's initialization stage, the DLRM model is trained on historical data generated by mobile  and  ubiquitous devices (such as cameras and sensors) on  edge workers. To ensure the continuous delivery of high-quality recommendations during the recommendation service, the DLRM model is continuously trained and refined online using streaming data from mobile devices~\cite{matam2024quickupdate, wang2024rap, yang2023incremental, sima2022ekko,yu2021secure,deng2024relayrec}. Overall, this paper focuses on online DLRM training using edge workers and parameter servers connected via a LAN to provide low-latency, privacy-preserving recommendation services for mobile and  ubiquitous applications.

\subsection{Architecture of Deep Learning Recommendation Model} 
Deep Learning Recommendation Models, such as WDL~\cite{cheng2016wide}, DFM~\cite{guo2017deepfm},  DCN \cite{wang2017deep} and FLEN~\cite{chen2019flen}, model the recommendation decision as a problem to predict the probability of a specific event, \eg, the likelihood of a web viewer watching the recommended content (video or music).  The general architecture of DLRM mainly consists of an embedding layer, feature interaction layer, and multi-layer perceptron layer (MLP), as shown in Fig.~\ref{fig:dlrmar}. The input of DLRM generally includes both dense and sparse inputs. Dense inputs represent continuous variables, such as age and height, while sparse inputs represent categorical variables, such as user ID and music ID. These are processed by the MLP layer and embedding layer, respectively.

\begin{figure}[hbt]
    \centering
    \includegraphics[width=0.66\columnwidth]{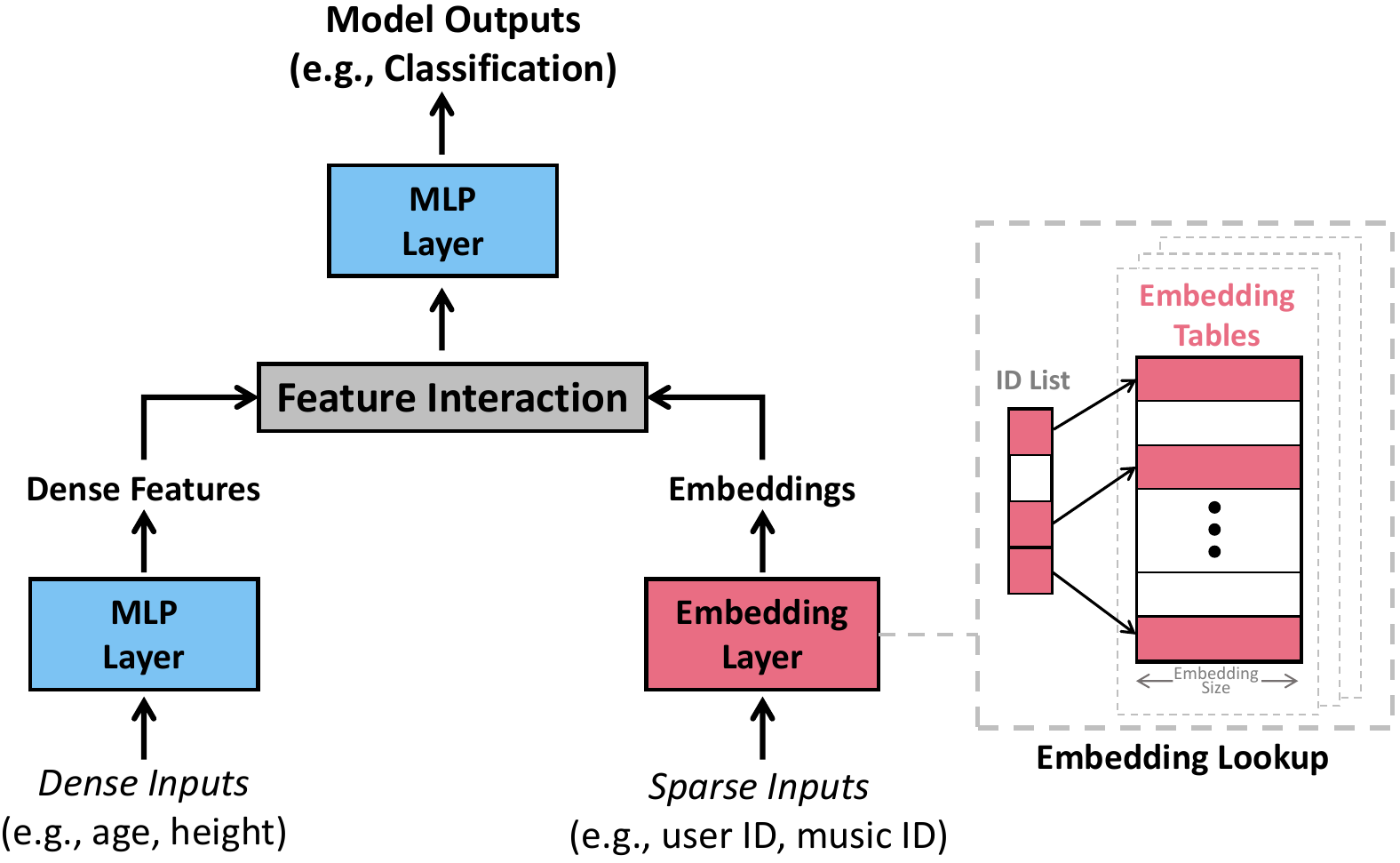}
    \caption{General Architecture of Deep Learning Recommendation Model.}
    \label{fig:dlrmar}
\end{figure}
\stitle{Embedding layer.} The embedding layer is practically the embedding tables, where each embedding table represents a categorical feature and each row of this table represents a specific ID (\eg, user ID or music ID). The embedding table transforms each ID into a fixed size vector with float values that are trainable, \ie, embedding.

\stitle{MLP layer.}  The multi-layer perceptron (MLP) layer performs multiple fully connected operations to process dense inputs and handle outputs from the feature interaction layer. It provides a probability of click-through rate (CTR) for a single user-item pair and includes other computation-intensive units, such as batch normalization.

\stitle{Feature interaction layer.} Feature interaction is the process of combining dense and sparse features to capture their complex and nonlinear relationships~\cite{zhag2022picasso}. The outputs of feature interaction will be processed by the subsequent MLP layer. Operations of feature interaction include averaging, sum, concatenation, \etc.

\subsection{Embedding Transmission Cost in DLRM Training} 

In DLRM training,  different parallelization and communication approaches can be implemented for the embedding and MLP layers. The MLP layers size is small enough to be accommodated by each worker, thus they are replicated across the workers to leverage data parallelism and use AllReduce for communication. On the other hand, embedding layers contribute to $>99\%$ of the entire model size and can be in the order of several terabytes. Because storing all the embedding tables in a single worker is not feasible,  the PS architecture maintains globally shared embeddings in memory-optimized parameter servers, and edge workers cache part of  embedding tables in their local memory, \ie, using model parallelism and PS communication~\cite{adnan2024heterogeneous}. However, this design introduces significant embedding transmission cost when applying  bulk synchronous parallel (BSP) training for non-degraded and reproducible model accuracy~\cite{rong2020distributed}.  BSP training is a parallel computing model that organizes computations into synchronized steps, where each step consists of local computation, communication among workers, and a global synchronization barrier. This ensures all workers stay coordinated and progress together~\cite{ghadimi2016mini}. In each training iteration, workers pull embeddings from the parameter server and push embedding gradients back to the parameter server on demand. In our experiment using the default setting as shown in Sec.~\ref{sec:eva}, the embedding transmission cost dominates the training cycle, which can account for up to $90\%$ of the end-to-end training time. Prior works mitigate embedding transmission cost by reducing the number of embedding transmissions between workers and parameter server~\cite{agarwal2023bagpipe,miao2021het,adnan2021accelerating,zeng2024accelerating} or reducing each embedding transmission cost between devices~\cite{song2023ugache,xie2023petps,wei2022gpu,kwon2022training}. In this paper, we aim to minimize the total embedding costs during DLRM training in edge environments when facing challenges such as cost composition, heterogeneous networks, and limited resources.
This work is similar to the former and orthogonal to the latter.

\section{System Model and Problem Formulation}\label{sec:formulation}

We provide the system model and problem formulation in this section. Commonly used  symbols are listed in Table~\ref{table:notation}.

\stitle{\textit{System.}} Motivated by training DLRM on multiple edge workers, this work focuses on a system comprising multiple edge GPU workers (workers) and one parameter server (PS). 
Specifically, the system consists of $n$  workers, $\mathcal{W}=\{w_1,w_2,\dots,w_{n}\} $. These workers are connected to the parameter server through Ethernet, typically $5$ Gbps or $0.5$ Gbps, and are also connected  among themselves.

\stitle{\textit{Input  Training Samples.}} In DLRM training, each worker is associated with a data loader, which is a process on the CPU and is responsible for preprocessing the training samples. The data loaders can load training samples from  local memory, disk, or network streams. For each training iteration $I_i$,  input generation is usually performed by the data loader, which partitions a batch of inputs into multiple micro-batches for every worker. In this paper, in order to maintain a balanced workload distribution across the workers, we use a fixed $m$ to present the batch size for each worker; thus, for $n$ workers, the total number of input training samples per iteration is $m \times n$.
For the  $m \times n$  input embedding samples $\mathcal{E}$, $\forall E_i \in \mathcal{E}$ is an embedding sample, consisting of multiple embedding IDs, $E_i = \{x_1, x_2,x_3, \dots\}$. The ID type feature is the sparse encoding of large-scale categorical information. We use $Emb(x_i)$  to represent the embedding value  for $x_i$ .

\stitle{\textit{Embedding Cache.}}
When training DLRM  on edge workers, each worker holds a replica of the MLP  and uses AllReduce for gradient synchronization during training. The entire embeddings are stored in the global embedding table on the parameter server, while each worker caches a portion of the embedding tables. The workers manage the local cache by transmitting embeddings with the PS to control inconsistencies between the local and global embedding tables. We use $r$ to represent the cache ratio, \ie, the ratio of the number of in-cache embeddings to the total number of embeddings.


\begin{table}[htbp]
\centering
\caption{List of Symbols}
\label{table:notation}
\begin{tabular}{>{\centering\arraybackslash}p{1cm} p{6cm}}
\toprule
\multicolumn{1}{c}{\textbf{Notation}} & \multicolumn{1}{c}{\textbf{Description}} \\
\midrule
$\mathcal{W}$  & The set of edge servers, $\mathcal{W}=\{w_1,w_2,\dots,w_{n}\}$ \\
$\mathcal{E}_i$  & The input embedding samples for iteration $I_i$, $\mathcal{E}_i=\{E_1,E_2,\dots,E_{m \times n}\}$ \\
$I_i$ & The $i_{th}$ training iteration \\
$m$ & Batch size for each worker, \ie, batch size per worker \\
$E_i$ & An embedding sample, $E_i=\{x_1,x_2,x_3, \dots\}$ \\
$x_i$ & The identifier (ID) for an embedding \\
$Emb(x_i)$ & The embedding vector corresponding to ID $x_i$\\
$ D_{tran}$ & The data size of an embedding vector\\
$B_w^j$ & The network bandwidth between worker $w_j$ and the parameter server\\
$T_{tran}^j $ & The cost for transmitting an embedding between worker $w_j$ and the parameter server, $T_{tran}^j = \frac{ D_{tran}}{B_w^j}$ \\

\bottomrule
\end{tabular}
\end{table}
\stitle{\textit{On-demand Synchronization Between Workers and the PS.}}
In BSP training for DLRM, each iteration requires the latest version of embeddings. To maintain consistency, workers transmit embedding values and gradients with the PS \footnote{In this paper, we refer to both embedding value and gradient transmission as embedding transmission, which includes embedding value pull and embedding gradient push.}.  On-demand synchronization is one method to reduce transmission overhead. Specifically, after iteration $I_t$, instead of pushing all embedding gradients to the PS, the on-demand push is used in $I_{t+1}$. For example, if embedding $Emb(x_i)$ is trained on worker $w_j$ during iteration $I_t$ and no other worker ($w' \in \mathcal{W}, w' \neq w_j$) receives samples containing $x_i$ in $I_{t+1}$, then worker $w_j$ does not need to push the gradient of $Emb(x_i)$ to the PS. Conversely, if other workers require $Emb(x_i)$ in $I_{t+1}$, worker $w_j$ needs to push $Emb(x_i)$ to the PS, and the other workers  pull $Emb(x_i)$ into their caches. Thus, when using on-demand synchronization, the process for each iteration on each worker after receiving training samples from the data loader is as follows:  on-demand embedding gradient pushing, pulling latest embeddings not in the local cache, model forward propagation, backward propagation, and dense parameter AllReduce synchronization. 


\stitle{\textit{Model Consistency Analysis.}}
In this part, we demonstrate that using on-demand synchronization in BSP does not compromise model accuracy.  Accelerating DLRM training under BSP with embedding dispatching can preserve model consistency and thus the exactly same model accuracy. Compared to vanilla distributed training, embedding dispatching introduces two main changes: a specific input embedding samples dispatch mechanism (versus a random dispatch mechanism) and on-demand synchronizations (versus full-set synchronizations). In BSP, on-demand synchronizations ensure that all workers use the latest embedding for the incoming training iteration, which is the same behavior as the vanilla distributed training performs.
We prove that the dispatch mechanism under BSP will not affect the training model. Considering a parameter optimizer followed by the SGD algorithm~\cite{goyal2017accurate}, the gradient calculation for model weights $y$ on a given batch of $m \times n $ training samples is expressed as follows:

\begin{equation}\label{eq:sgd}
\nabla_{y} = \frac{1}{m \times n } \sum_{i=1}^{m \times n } \frac{\partial L(\mathcal{E}_i, y)}{\partial y}
\end{equation}

where \( \mathcal{E}_i \) represents the \( i \)-th training sample of the batch, and \( L \) denotes the loss function. According to Eq.~\ref{eq:sgd}, the gradient for the batch is calculated as the sum of the individual gradients for each training sample. Since the individual gradient depends only on input samples and current model weights, which have been synchronized before this iteration begins under BSP, partitioning the batch into \( n \) micro-batches does not affect the gradient result:

\begin{equation}
\frac{1}{m \times n } \sum_{i=1}^{m \times n } \frac{\partial L(\mathcal{E}_i, y)}{\partial y} = \frac{1}{m \times n } \sum_{i=1}^{n} \sum_{j=1}^{m} \frac{\partial L(E_{ij}, y)}{\partial y}
\end{equation}

where \( E_{ij} \) is the \( j \)-th training sample in the \( i \)-th micro-batch (for worker \( w_i \)). Therefore, any dispatch result generated by any dispatching mechanism preserves the same gradients as in BSP throughout training and converges to the same model.
\subsection{System Model}\label{sec:systemmodel}
\begin{figure*}[t]
    \centering
    \includegraphics[width=1.99\columnwidth]{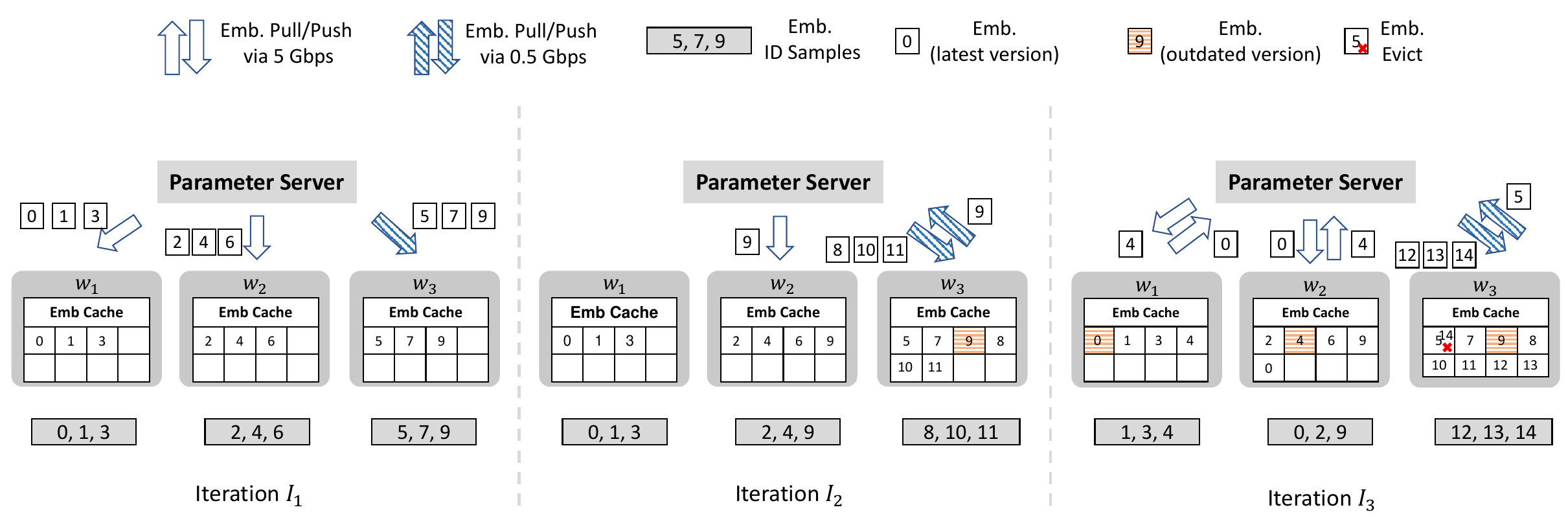}
    \caption{Miss Pull, Update Push, and Evict Push transmission operations in DLRM training.}
    \label{fig:example}
\end{figure*}

\subsection{Problem Formulation}
In this work, our goal is to minimize the total embedding transmission cost when training DLRM in edge environments. As shown in Fig.~\ref{fig:example}, in DLRM training, when using on-demand synchronization, the cost for embedding transmission comes from three embedding transmission operations, \ie, \textit{Miss Pull}, \textit{Update Push}, and \textit{Evict Push}. In Fig.~\ref{fig:example}, for simplicity, numbers are used to represent embedding IDs and vectors.

\begin{itemize}[leftmargin=*]
    \item \textit{Miss Pull}: When receiving the training samples from the data loader, the worker should retrieve the necessary embeddings from the local cache or the PS, a process known as embedding lookup. If the required embedding  is available with the latest version, there is no need to pull the embedding from the PS. In the context of cache, this is a cache hit. Conversely, if the required embedding is absent  or  outdated, it must be pulled from the PS, resulting in a cache miss and incurring additional transmission cost (Miss Pull), as shown for embeddings $Emb(x_8)$, $Emb(x_9)$, $Emb(x_{10})$, and $Emb(x_{11})$ in iteration $I_2$ in Fig.~\ref{fig:example}.

    \item \textit{Update Push}: When using on-demand synchronization, the gradient push occurs at the start of each iteration rather than at its end. Upon receiving the training samples, if some embeddings were trained on other workers in the previous iteration, the other worker should Update Push the embedding gradient to the parameter server. For example, in iteration $I_2$ in Fig.~\ref{fig:example}, worker $w_3$ should push the embedding gradient of $x_{9}$ to the PS for training on worker $w_2$.

    \item \textit{Evict Push}: In addition to the Miss Pull caused by cache misses and the Update Push introduced by on-demand synchronization, if the worker's cache is full, it needs to evict embeddings to make room for the current training iteration. During eviction, if unsynchronized embeddings are evicted, their gradients must be updated to the parameter server. This process is referred to as Evict Push (\eg, the eviction of $Emb(x_{5})$ in iteration $I_3$ in Fig.~\ref{fig:example}).
\end{itemize}

The goal of this paper is to minimize the total embedding transmission cost across all iterations. We assume that the data size for each embedding pull and push is the same, denoted as $D_{tran}$. For iteration $I_t$, the number of embedding transmissions is $T_{num}^t$. Given a network bandwidth of $B_w$, the cost for each transmission operation is $T_{tran} = \frac{D_{tran}}{B_w}$.
\begin{align} 
\intertext{Problem P:}
    \begin{aligned}
    \text{minimize }        &\sum_{\text{for all iterations}} T_{num}^t \times T_{tran} \label{problem1}\\
    \end{aligned}
\end{align}


\section{Mechanism Design}\label{sec:design}

In this section, we first demonstrate the process of input embedding samples dispatching in \our. Second, we introduce how to calculate the expected transmission cost in Sec.~\ref{sec:calculate}. Then, given the limited resources in edge environments, we propose \ourmix, a hybrid dispatch decision method that balances solution quality and resource consumption in Sec.~\ref{sec:hybrid}.

\subsection{Process of  embedding samples dispatching  in \our}

Fig.~\ref{fig:disoverview} presents an overview of the embedding samples dispatching process in \our.  Specifically, at the start of  training iteration $I_t$, \our makes dispatch decisions for iteration $I_{t+1}$ based on the input embedding samples for iteration $I_{t+1}$ and the current states (cached or not, the latest version or not) of the cached embeddings on each worker. Additionally, based on the dispatch decision and the current states of the cached embeddings on each worker, \our can generate the update pushing plan for each worker for iteration $I_{t+1}$, which is omitted in Fig.~\ref{fig:disoverview}.  Due to this paper focusing on online DLRM training, the dispatch decision time for iteration $I_{t+1}$ should be hidden within the training time of $I_t$ to ensure that embedding sample dispatching can reduce the training time. The above dispatch process is executed locally on each worker, rather than by a centralized orchestrator, to avoid the overhead of distributing dispatching decisions.

\begin{figure}[h]
    \centering
    \includegraphics[width=0.85\columnwidth]{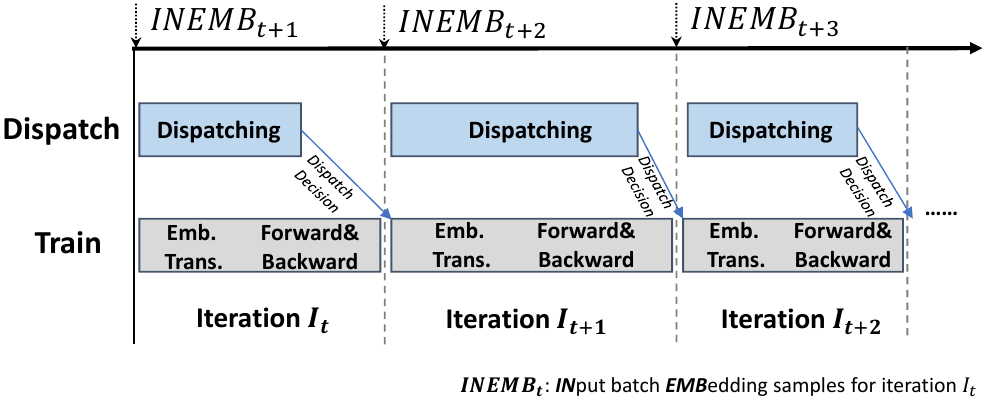}
    \caption{Overview of  embedding samples dispatching process in \our.}
    \label{fig:disoverview}
\end{figure}


To reduce the embedding transmission cost, for input embedding samples $\mathcal{E}_i$, we assign an expected transmission cost $c_i^j$ for each embedding sample $E_i$ ($E_i \in \mathcal{E}_i$) training on  worker $w_j$, and dispatch embedding samples to the workers to minimize the total embedding transmission cost.
When dispatching input embedding samples based on the expected transmission cost, decision-making consumes the resources of edge workers, which are limited and subject to competition from multiple tasks. \our incorporates a hybrid dispatch decision method \ourmix, which combines a resource-intensive optimal algorithm (\opt)  and a  heuristic algorithm (\heu) to balance solution quality (\ie, the total embedding transmission cost) and resource consumption. For example, \ourerwu indicates that $25\%$ of the input embedding samples use \opt for dispatch decisions, while the remaining $75\%$ use \heu. We analyze the dispatch error when using \heu, and partition  input embedding samples to different dispatch decision methods by $\min_{\text{2}}$ - $\min$, where $\min_{\text{2}}$ and $\min$ represent the second minimum and minimum costs, respectively, for training embedding sample $E_i$ on different workers.

\subsection{Expected Transmission Cost}\label{sec:calculate}
This subsection  demonstrates the method to calculate the expected transmission cost.

Since the  embedding states  of the embedding cache  on each worker are visible, and today's data loaders can prefetch  embedding samples for the next iteration, as shown from Line~\ref{costbegin} to Line~\ref{costend} in Alg.~\ref{alg:dispatch}, we calculate the expected transmission cost of dispatching each embedding sample to each worker. We introduce the calculation method through an example. Considering the start of iteration $I_t$, we begin calculating the expected transmission costs for  embedding samples for the next iteration, \ie, $I_{t+1}$. If the input embedding sample $E_i = \{x_{1}, x_{2}, x_{3}, \dots\}$ is dispatched to worker $w_j$ for training, for embedding value $Emb(x_i)$ of $x_{i}$, if the latest version of $Emb(x_i)$ is not in $w_j$'s cache, a Miss Pull transmission operation is required, with a cost of $T_{tran}^j$ (Line~\ref{costpull}). Additionally, if the latest version of embedding $Emb(x_i)$ exists on other workers during the current iteration $I_t$, under on-demand synchronization, the other worker ($w_{j'}$) needs to push $Emb(x_i)$ to the parameter server, thus $c_i^j += T_{tran}^{j'}$, \ie, the cost of worker $w_{j'}$ pushing the embedding to the parameter server is added (Line~\ref{costpush}). This way, for iteration $I_{t + 1}$, we obtain the expected transmission cost of dispatching each embedding sample to different workers and store the expected transmission cost in the cost matrix $\mathcal{C}$. For Line~\ref{pullin} and Line~\ref{pushin} in Alg.~\ref{alg:dispatch}, \textit{in} indicates that the latest version of \( Emb(x_i) \) is available in the embedding cache of the edge worker, while \textit{not in} means it is not.

Moreover, training DLRM in edge environments  faces heterogeneous network bandwidth between edge workers and the PS, \ie, different workers have varying network bandwidths ($B_w^j$), leading to different $T_{tran}^j$.  For instance, for iteration $I_2$  in Fig.~\ref{fig:example}, the bandwidth between $w_2$ and the PS is $5$ Gbps, whereas between $w_3$ and the PS, it is only $0.5$ Gbps. As a result, the cost to pull $Emb(x_8)$ from the PS to $w_3$ is ten times  the cost to pull $Emb(x_9)$ to $w_2$, \ie, $T_{\text{tran}}^{3} = 10 \times T_{\text{tran}}^{2}$.

\begin{algorithm}[h]
    \caption{\our}  
    \label{alg:dispatch}
    \KwIn{Input  embedding samples $\mathcal{E}$ for one iteration, $| \mathcal{E}| = m \times  n$, workers $\mathcal{W}$, cost matrix $\mathcal{C}$.}

    \KwOut{Dispatch decision (\textit{Decision}).}
    Initialize \textit{Decision};\ 
    
   Initialize all elements in $\mathcal{C}$ with $0$;\
   
    \For{$E_i$ in $\mathcal{E}$}{ \label{costbegin}
     \For{$w_j$ in $\mathcal{W}$}{
     \For {$x_{i}$ in $E_i$}{
        \If{$Emb(x_i)$ for $x_{i}$ not in $w_j$\label{pullin}} {
               $c_i^j += T_{\text{tran}}^j$;\ \label{costpull}
               
          \If{$Emb(x_i)$ for $x_i$ in $w_{j'}$}{ \label{pushin}
           $c_i^j += T_{\text{tran}}^{j'}$;\
           \label{costpush}
          }
        }}
     }
    } \label{costend}
    \textit{Decision} = \ourmix($\mathcal{C}$)\;
\Return $\textit{Decision}$.
\end{algorithm}

\subsection{\ourmix}\label{sec:hybrid}

After calculating the expected transmission cost for training each embedding sample on each worker and obtaining the cost matrix $\mathcal{C}$, this section focuses on designing a  decision method to determine the  dispatch decision that minimizes the total embedding transmission cost. 

In resource-limited edge environments, ensuring the time required to derive a dispatch decision remains shorter than the training time per iteration poses a  challenge.  We first identify the limitation of sequentially executing the Hungarian algorithm on  CPU, which results in excessive solution time. We then briefly analyze the feasibility of parallelizing the Hungarian algorithm using CUDA programming (\opt). To balance solution quality (\ie, the total embedding transmission cost) and resource consumption, we introduce a resource-efficient heuristic method \heu. While \opt delivers optimal dispatch decisions with high resource consumption, we propose a hybrid dispatch method \ourmix that combines  \opt and \heu. In \ourmix, the dispatch error of \heu is a partitioning criterion for assigning portions of the cost matrix to \opt and \heu.

The problem of finding the optimal dispatch decision of cost matrix $\mathcal{C}$ that minimizes the total embedding transmission cost can be solved as an assignment problem using the Hungarian algorithm~\cite{kuhn1955hungarian}. When the batch size per worker is $m$, and the number of workers is $n$, the cost matrix $\mathcal{C}$ has dimensions of $(m \times n) \times n$. To accommodate the constraint of $m$ batches per worker, we expand each column of $\mathcal{C}$ to $m$ columns, resulting in a square matrix $\mathcal{C^{'}}$ with order $k$ ($k = m \times n$), which serves as input to the Hungarian algorithm. The Hungarian algorithm has a time complexity of $\mathcal{O}(k^{3})$~\cite{lawler2001combinatorial}.  As demonstrated in Table~\ref{tab:execution_times}, sequential execution on CPU shows that when the batch size per worker increases from $32$ to $1024$, the execution time escalates from $9$ milliseconds to $134.986$ seconds. This significant latency exceeds the training time per iteration, violating our requirement that the dispatch decision for iteration $I_{t+1}$ should be computed within the training time of iteration $I_t$ to avoid increasing the end-to-end training time. Our analysis reveals that the Hungarian algorithm contains several inherently parallel components, including subtracting minimum values from each row or column in initial reduction, zero-element covering in initial matching, and matrix updates during matrix adjustment~\cite{lopes2019fast}. Therefore, we implement a CUDA-based parallel version of the Hungarian algorithm to ensure the dispatch decision latency remains within the training time constraints. The effectiveness of this parallelization is evidenced by the execution times shown in Table~\ref{tab:execution_times}. For a comprehensive description of the Hungarian algorithm, please refer to~\cite{lawler2001combinatorial,munkres1957algorithms}.

\begin{table}[h]
    \centering
      \caption{Execution time (millisecond) comparison between serial and parallel implementations of Hungarian algorithm for different batch size  per worker (BPW) when using $8$ workers.}
    \begin{tabular}{ccccccc}
        \toprule
         \textbf{Batch Size  Per Worker}& \textbf{32} & \textbf{64} & \textbf{128} & \textbf{256} & \textbf{512} & \textbf{1024} \\ 
        \midrule
        Serial & 9 & 62 & 528 & 3360 & 50976 & 134986 \\ 
        Parallel & 21 & 28 & 82 & 186 & 811 & 1385 \\ 
        \bottomrule
    \end{tabular}
  
    \label{tab:execution_times}
\end{table}

\begin{algorithm}[hbt]
    \caption{\ourmix}  
    \label{alg:hybrid}
    \KwIn{Input embedding samples $\mathcal{E}$, $| \mathcal{E}| = m \times n$, workers $\mathcal{W}$,  \textit{maxworkload} = $m$, cost matrix $\mathcal{C}$,  $\alpha$, workload list $workload[1 \ldots n]$.}
    \KwOut{$\mathcal{D}$. \tcp*[h]{Dispatch Decision}\\}
    
    Initialize  $workload[1 \ldots n]$ with zeros;\
    
    Calculate  $\min_{\text{2}}$ - $\min$ of each row  in $\mathcal{C}${\label{min2start}} ;\
    
    Sort rows   in $\mathcal{C}$ by $\min_{2}$ - $\min$  in descending order;\

    \( \mathcal{C}_{\text{heu}} \gets \) rows of \( \mathcal{C} \) from \( \lceil | \mathcal{E}| \times \alpha \rceil  \) to \( | \mathcal{E}| \)\;\tcp*[h]{Preserve initial row numbering}

    \(\mathcal{C}_{\text{opt}} \gets \) rows of   \(\mathcal{C}\) from $0$  to \( \lfloor | \mathcal{E}| \times \alpha \rfloor \) {\label{min2end}};\
   
    Expand each column to \( \lfloor \text{\textit{maxworkload}} \times \alpha\rfloor \) columns of $\mathcal{C}_{\text{opt}}$;\

   $\mathcal{D}$ = \opt($\mathcal{C}_{\text{opt}}$);\
    
    \textit{maxworkload} =  $\text{\textit{maxworkload}} - \lfloor \text{\textit{maxworkload}} \times \alpha\rfloor $;\

    \For{ $row_i$  in $\mathcal{C}_{\text{heu}}$ }{\label{heustart} 
\While {$\mathrm{True}$}{
 Select the minimum value $c_i^j$ in  $row_i$ ;

 \If{$workload[j]$ $<$ \text{\textit{maxworkload}}}{

Dispatch $E_i$ to worker $w_j$;\

Add dispatch decision $E_i$ to $w_j$ into    $\mathcal{D}$;\

 $workload[j]$ += 1;\
 
 break;\
 
 }

 \Else{

  Exclude $c_i^j$ from the  selection in  $row_i$;\
 }
 }

    }\label{heuend}

\Return   $\mathcal{D}$.
\end{algorithm}

In edge environments, although the Hungarian algorithm can provide the optimal dispatch decision to minimize the total embedding transmission cost for each iteration,  implementing the Hungarian algorithm using CUDA consumes inherently limited and competing GPU resources. To balance the dispatch decision quality and the resource consumption, as shown in Alg.~\ref{alg:hybrid}, we propose \ourmix, a hybrid method to dispatch the input embedding samples to edge workers, which combines the Hungarian algorithm (\opt) and the heuristic algorithm (\heu) to  make dispatch decision.

 Line~\ref{heustart} to Line~\ref{heuend} in Alg.~\ref{alg:hybrid} show the heuristic method \heu. For dispatching  input embedding samples to  workers, since we have  the cost matrix $\mathcal{C}$, where $row_i$ represents $n$ expected embedding cost for dispatching $E_i$  to $n$ workers, \heu greedily dispatches $E_i$ to the worker with the lowest expected cost unless the worker reaches its maximum workload. For example, if $c_i^j$ is the minimum value in $row_i$ of $\mathcal{C}$, then $E_i$ is dispatched to $w_j$. In this paper, for each iteration, to avoid imbalance among different workers, each worker processes  $m$ embedding samples, \ie, $m$ is the   \textit{maxworkload} of each worker. While this approach is not optimal, according to Theorem~\ref{the:error}, the worst-case dispatch error is $\min_{\lfloor i/m \rfloor +1 }-\min$ for $row_i$, $\min_{\lfloor i/m \rfloor +1}$ and $\min$ represent the $\lfloor i/m \rfloor +1$-th minimum and the minimum values in $row_i$, respectively.

To balance decision quality and computational efficiency, \ourmix employs the $\min_{\text{2}}-\min$ metric as the partitioning criterion for each row. Specifically, we assign a fraction $\alpha$ ($0 \le \alpha \le 1$) of rows with the highest $\min_{\text{2}}-\min$ values to \opt, while the remaining rows are processed by \heu (Line~\ref{min2start} to Line~\ref{min2end}). This ensures that embedding samples with larger potential dispatch errors are handled by the optimal solver. It is worth noting that the partitioning criterion is flexible and can be adapted based on specific requirements. Alternative metrics such as $\min_{\text{3}}-\min$, $\min_{\text{3}}-\min_{\text{2}}$, or row-wise averages can be employed, depending on the data distribution patterns observed in practical DLRM training scenarios.


\begin{theorem}\label{the:error}
When using \heu, for $row_i$, the worst-case dispatch error is $\min_{\lfloor i/m \rfloor+1}-\min$. $\min_{\lfloor i/m \rfloor+1}$ and $\min$ represent the $\lfloor i/m \rfloor+1$-th minimum and the minimum values  in $row_i$, respectively.
\end{theorem}

\begin{proof}
For the cost matrix $\mathcal{C}$, with $m \times n$ rows and $n$ columns, each column can be dispatched up to $m$ times. For $row_i$ ($0 \le i \le m-1$), the worst-case error is 0, as no column will reach  \textit{maxworkload}. For $row_i$ ($m \le i \le 2m-1$), the worst-case scenario is that all $row_i$ ($0 \le i \le m-1$) are dispatched to the same worker $w_{u^1}$, and in $row_i$ ($m \le i \le 2m-1$), the minimum values are all $c_i^{u^1}$. Since the column representing worker $w_{u^1}$ has already reached  \textit{maxworkload}, for $row_i$ ($m \le i \le 2m-1$), we can only choose $\min_{2}$. For $row_i$ ($2m \le i \le 3m-1$), similarly, if all $row_i$ ($0 \le i \le m-1$) are dispatched to the same worker $w_{u^1}$ and all $row_i$ ($m \le i \le 2m-1$) are dispatched to the same worker $w_{u^2}$, for $row_i$ ($2m \le i \le 3m-1$), their minimum values and $\min_{2}$ are both located in $c_i^{u^1}$ or $c_i^{u^2}$. But since $w_{u^1}$ and $w_{u^2}$ have both reached  \textit{maxworkload}, $row_i$ ($2m \le i \le 3m-1$) can only choose $\min_{3}$. By this logic, for $row_i$ (\ie, embedding sample $E_i$), the worst-case dispatch error is $\min_{\lfloor i/m \rfloor+1}-\min$. $\min_{\lfloor i/m \rfloor+1}$ and $\min$ represent the $\lfloor i/m \rfloor+1$-th minimum and the minimum values (costs) in $row_i$, respectively.

\end{proof}

\section{Implementation}\label{sec:imple}

We  implement a  prototype of \our on top of HET~\cite{miao2021het} with C++  and Python. Besides the design described in Sec.~\ref{sec:design}, \our maintains cache snapshots of all workers  to provide information on embedding locations. These cache snapshots are used to calculate expected transmission cost, and updated at the end of the dispatch based on the dispatch result. We replace the data loader implementation in HET, the data loader in \our  returns embedding samples and on-demand synchronized embedding indexes as sparse inputs. To support input prefetching without interfering with the training process, \our is launched by dedicated threads, and the dispatching decisions of \our are transmitted to the data loader via shared memory. Additionally, we use CUDA to parallelize the initial reduction, initial matching, and matrix adjustment steps of the Hungarian algorithm.

\section{Evaluation}\label{sec:eva}

We evaluate the performance of \our using typical workloads as listed in Table~\ref{tab:workloads}. Compared with LAIA, the state-of-the-art mechanism that reduces the training time by reducing the number of embedding transmissions, \our achieves up to a $1.74\times$ speedup and reduces embedding transmission cost by up to $36.76\%$. Through sensitivity analysis on the batch size per worker, cache ratio, and embedding size, \our consistently outperforms baselines. All experiments are conducted five times, and the average results are reported.

\subsection{Evaluation Settings}

\stitle{Testbed.}
Both workers and PS are equipped with a 28-core Intel Xeon Gold 6330 CPU at $2.0$ GHz and $64$ GB ($500$ GB in PS) of RAM. Each worker is  equipped with a Nvidia $4090$ GPU. These workers and the PS are connected via $5$ or $0.5$ Gbps Ethernet. All machines run Ubuntu 18.04, Python 3.8, CUDA 11.3, cuDNN 8.2.0, NCCL 2.8, and OpenMPI 4.0.3.

Unless mentioned otherwise, our experimental setup consists of $8$ edge workers and $1$ PS. Four workers are connected to the PS via $5$ Gbps Ethernet, and the other four via $0.5$ Gbps Ethernet.  The batch size per worker is $128$ and the embedding size is $512$. Each worker contains an embedding cache  that can house $8\%$ of the PS-side embedding tables by default.

\stitle{Baselines.}
We compare \our with the following mechanisms.

\begin{itemize}[leftmargin=*]
\item HET~\cite{miao2021het}: HET is a mechanism that enables embedding caching under the PS architecture by tracking embedding versions to tolerate bounded staleness. When retrieving an embedding, the local version is compared with the PS version, and if the difference exceeds a threshold, the embedding is pulled from the PS. This same method applies to gradient synchronization.
\item FAE~\cite{adnan2021accelerating}: FAE is designed with static caching. Cached embeddings are profiled and fixed offline before training. All workers in FAE cache  the same embeddings, and they synchronize all cached embeddings using AllReduce.
\item LAIA~\cite{zeng2024accelerating}: LAIA is designed for scheduling embeddings among multiple cloud workers. LAIA calculates a score to quantify the relevance between every input sample and worker and allocates each input to the worker with the highest score.
\end{itemize}

 HET and FAE reduce the embedding transmissions by compromising model accuracy and are orthogonal to \our, and we adopt BSP training in HET and FAE.  In this section, \our is used to represent the mechanism regardless of the value of $\alpha$. \ouryi, \ourwu, \ourerwu, \ouryierwu, and \ourling represent the  mechanisms with different $\alpha$.

\stitle{Workloads.} We use three real-world models with representative datasets for end-to-end experiments, as listed in Table~\ref{tab:workloads}. We exclude the first $10$ iterations for warm-up and report the performance for the subsequent iterations. We do not include training accuracy results because \our can preserve a consistent model accuracy, as analyzed in Sec.~\ref{sec:systemmodel}. 

\begin{table}[hbt]
  \caption{Workloads for evaluation.}
    \centering
     \small 
    \begin{tabular}{ccc}
        \toprule
     \textbf{Workload}&   \textbf{Model} & \textbf{Dataset} \\ \midrule
        S1 & WDL \cite{cheng2016wide} & Criteo Kaggle \cite{criteodata} \\ 
        S2 & DFM \cite{guo2017deepfm} & Avazu \cite{avazudata} \\ 
        S3 & DCN \cite{wang2017deep} & Criteo Sponsored Search \cite{tallis2018reacting} \\ 
        \bottomrule
    \end{tabular}
   
    \label{tab:workloads}

\end{table}


\stitle{Metrics.} The metrics used to evaluate the performance of the mechanisms are  \underline{It}erations  \underline{p}er  \underline{S}econd (ItpS), and the total embedding transmission  \underline{Cost} (Cost). To facilitate comparisons, the ItpS of LAIA is used as the reference for describing, with the performance improvements of other mechanisms expressed as a ratio to this reference.
\begin{equation*}
    \text{Speedup of A} = \frac{\text{ItpS(A)}}{\text{ItpS(LAIA)}}.
\end{equation*}
LAIA is also used as a reference for embedding transmission cost reduction.
\begin{equation*}
    \text{Cost Reduction of A} = \frac{\text{Cost(LAIA)-Cost(A)}}{\text{Cost(LAIA)}}.
\end{equation*}

\subsection{Overall Performance}\label{subsec:overall}

We first evaluate the overall performance of speedup and transmission cost reduction of \ouryi, \ourwu and \ourling and compare it with baselines under  the default setting. Fig.~\ref{fig:overall} shows the training speedup of \our \ouryi, \ourwu and \ourling over baselines across different workloads. We observe that \our can achieve $1.03\times$ to $1.74\times$ speedup compared to LAIA (used as the reference and not shown in Fig.~\ref{fig:overallperformance}). Specifically, we find that as $\alpha$ decreases, the speedup also decreases. This is reasonable because a larger $\alpha$ means more embedding samples are handled by \opt, resulting in lower embedding transmission cost per iteration. As shown in Fig.~\ref{fig:overallcost},  \ouryi reduces transmission cost by up to $36.76\%$ compared to LAIA, while \ourwu and \ourling achieve $10.81\%$ and $7.03\%$ reductions, respectively. Our experimental results indicate that FAE and HET consistently underperform compared to LAIA. Therefore, to focus on the improvements of \our over LAIA, we omitted the results of FAE and HET in subsequent experimental analyses.
\begin{figure}[h]
\centering
\includegraphics[width=0.8\linewidth]{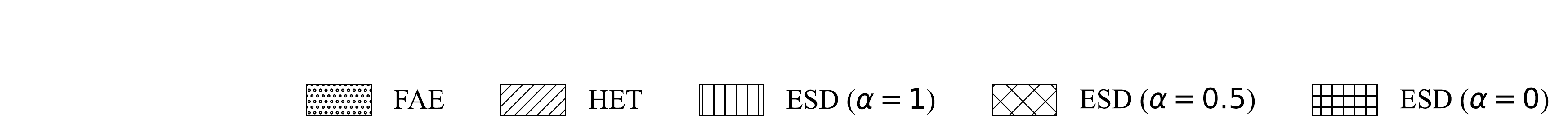}

\vspace{10pt} %

\begin{subfigure}{0.45\linewidth}
\centering
\includegraphics[width=\linewidth]{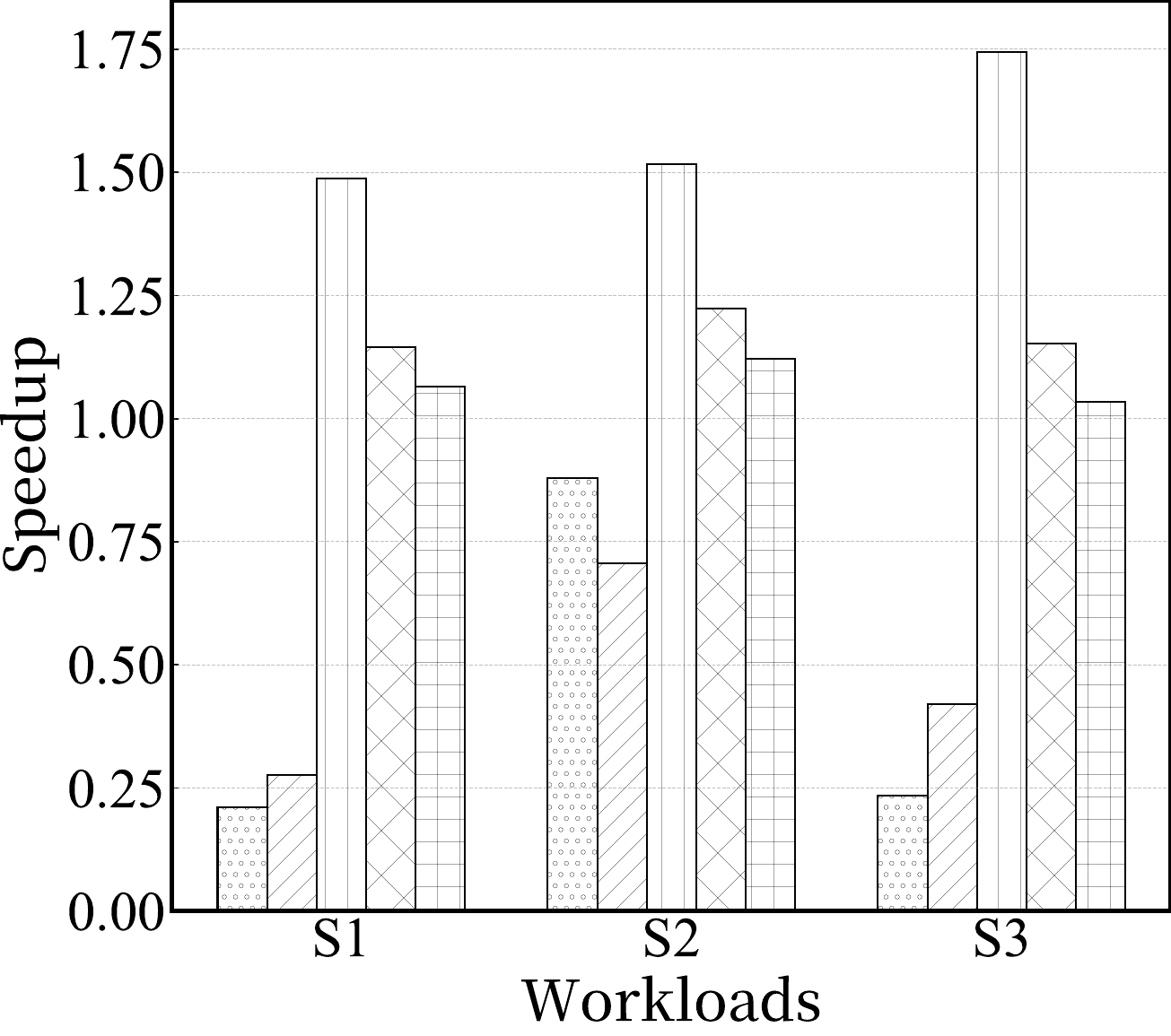}
\caption{ItpS Speedup}
\label{fig:overall}
\end{subfigure}
\hfill
\begin{subfigure}{0.45\linewidth}
\centering
\includegraphics[width=\linewidth]{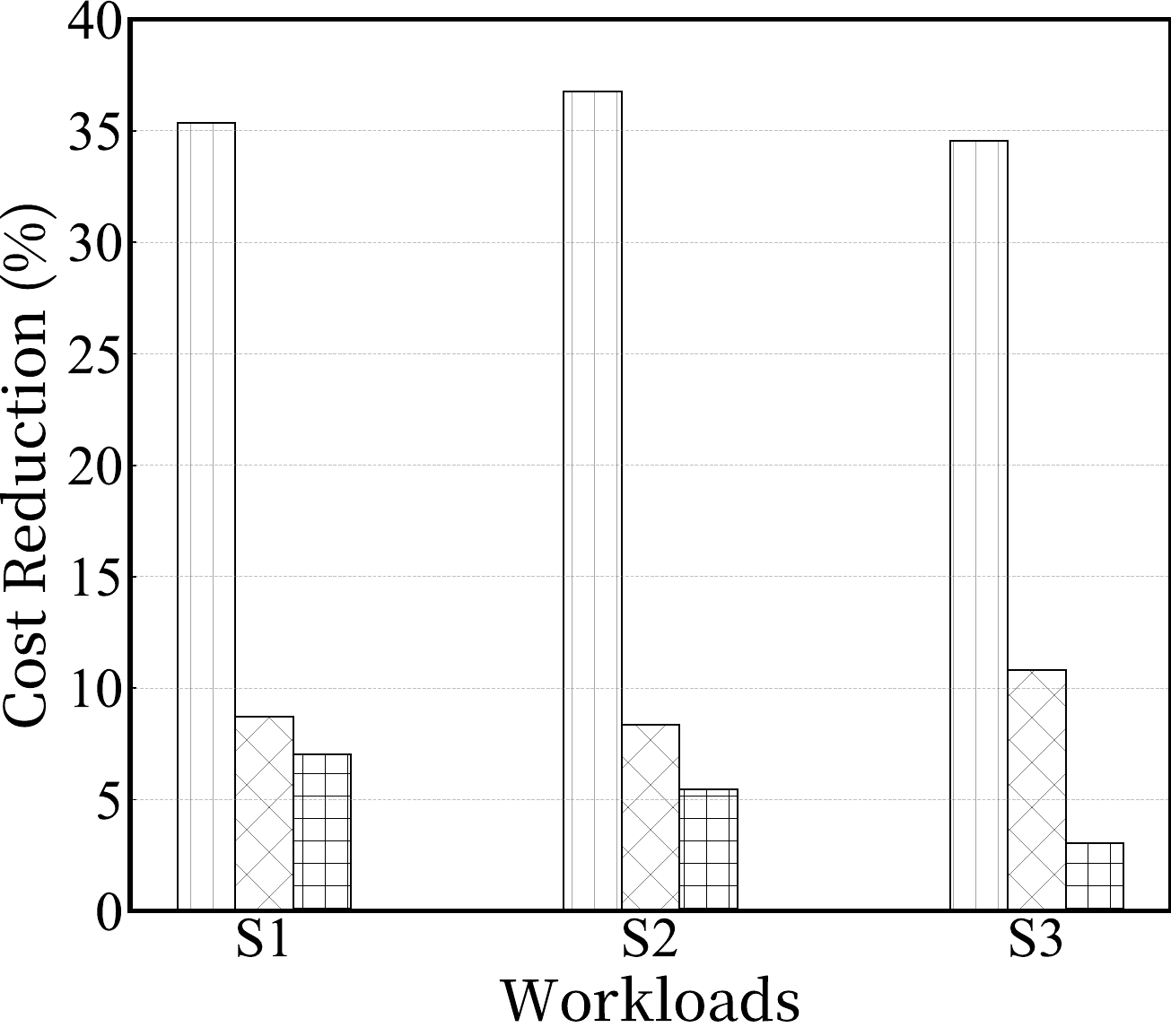}
\caption{Cost Reduction}
\label{fig:overallcost}
\end{subfigure}

\caption{Overall performance.} 
\label{fig:overallperformance}
\end{figure}



\subsection{Hit Ratio and Ingredient of Transmission Operations }\label{eva:ingredient}

To explore the reasons that affect the performance of these mechanisms, we show the hit ratio and  the ingredient of embedding transmission operations in Fig.~\ref{fig:hitandcost}. 

When we assume \(x_i\) exists in the input embedding samples for worker \(w_j\), if the latest version of \(Emb(x_i)\) is already cached in \(w_j\)'s embedding cache, it is considered a \textit{hit}. The hit ratio is defined as the fraction of embedding lookups that hit the cache relative to the total number of lookups.  We present the hit ratio for LAIA, \ouryi, \ourwu, and \ourling in Fig.~\ref{fig:hit}. Compared to LAIA, \our does not achieve a higher hit ratio. However, as shown in Fig.~\ref{fig:overallperformance}, \our reduces the cost and accelerates end-to-end training compared to LAIA. This is expected, as discussed in Sec.~\ref{sec:intro}, since embedding transmission cost arises not only from embedding misses but also from update push and evict push operations. Therefore, a higher hit ratio does not necessarily lead to reduced embedding transmission.

\begin{figure}[h]
\centering

\begin{subfigure}{0.48\linewidth}
\centering
\includegraphics[width=\linewidth]{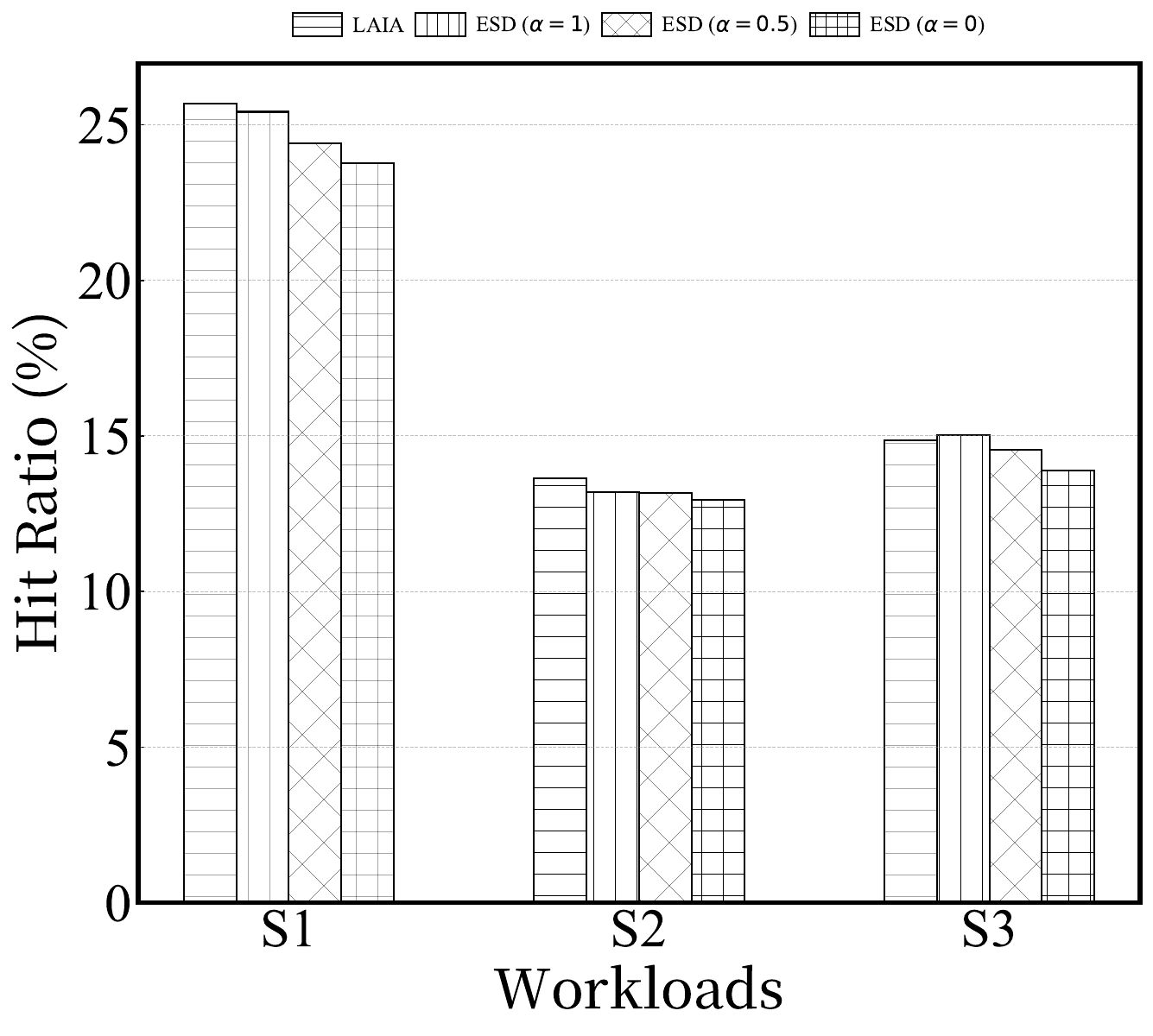}
\caption{Hit ratio}
\label{fig:hit}
\end{subfigure}
\hfill
\begin{subfigure}{0.48\linewidth}
\centering
\includegraphics[width=\linewidth]{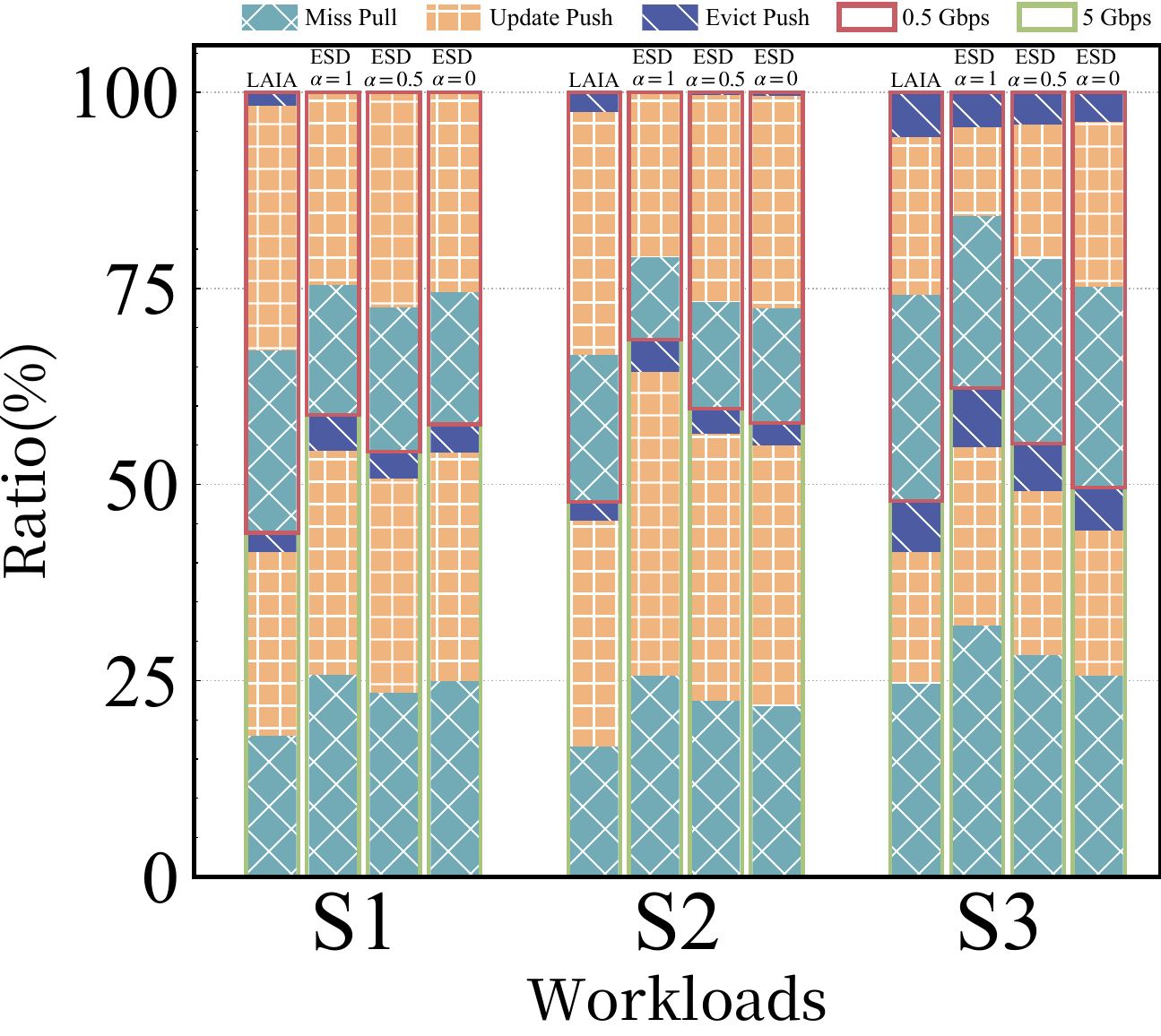}
\caption{Ingredient of trans. op.}
\label{fig:ingrecost}
\end{subfigure}

\caption{Hit ratio and Ingredient of transmission operations.} 
\label{fig:hitandcost}
\end{figure}

Fig.~\ref{fig:ingrecost} illustrates the ingredient of embedding transmission operations for different mechanisms, \ie,  the proportions of miss pull, update push, and evict push operations in the total transmission operations.  In each bar chart, the bottom three sections represent the results for 5 Gbps workers, while the upper sections correspond to 0.5 Gbps workers\footnote{For simplicity, ``5 Gbps worker(s)'' refers to workers connected to the parameter server via 5 Gbps Ethernet, and ``0.5 Gbps worker(s)'' similarly refers to workers connected via 0.5 Gbps Ethernet.}. As Fig.~\ref{fig:ingrecost} shows, compared to LAIA, \our increases the proportion of operations on 5 Gbps workers across all three workloads, whereas in LAIA, the proportion of operations on 5 Gbps workers is smaller than that on 0.5 Gbps workers. This demonstrates that \our effectively considers the heterogeneous networks.
Fig.~\ref{fig:ingrecost} also shows that miss pull and update push account for over 90\% of embedding transmission operations, while evict push contributes less than 10\%.

\subsection{Cost Reduction and  GPU Resource Consumption}

Due to limited resources on edge workers, this paper proposes \ourmix, a hybrid method to make the input embedding samples dispatch decision, with $\alpha$ representing the proportion dispatch using \opt. As shown in Fig.~\ref{fig:gputil}, when $\alpha$ is $1$, $0.5$, $0.25$, $0.125$, and $0$, with batch sizes per worker of $128$ and $256$, we demonstrate the cost reduction and the  GPU resource consumption (represented by GPU Utilization). For \ourling, the GPU utilization is $0$. When the batch size per worker is $128$, as shown in Fig.~\ref{fig:gputil128}, under different workloads, a larger $\alpha$ results in greater cost reduction and higher GPU utilization. On S1,  \ouryi achieves the maximum cost reduction, and the GPU utilization is $74.89\%$. When the batch size per worker is $256$, as shown in Fig.~\ref{fig:gputil256}, \our can reduce the embedding transmission cost by up to $40.84\%$. For \ourwu, when the batch size per worker is $128$, GPU utilization ranges from $17.59\%$ to $19.72\%$ across three workloads, while for a batch size of $256$, GPU utilization ranges from $50\%$ to $59\%$. Overall, a larger $\alpha$ results in greater cost reduction and higher GPU utilization. The setting of $\alpha$ depends on the limited GPU resources and the competition among multiple tasks on edge workers. However, even with $\alpha = 0$, \our can still reduce the transmission cost relative to LAIA.

It is important to note that accurately measuring GPU resource consumption is a challenging issue in computer systems~\cite{delestrac2024multi,shubha2024usher}. In Fig.~\ref{fig:gputil}, this paper uses the command \texttt{nvtop} to obtain the GPU utilization, which only provides a coarse-grained measure of GPU utilization. Moreover, GPU utilization of \our in Fig.~\ref{fig:gputil} is measured by executing \our independently on the GPU rather than during the training process. In fact, when training on workload S3 with batch size per worker is $128$, the GPU utilization observed via \texttt{nvtop} is approximately 55\% when \ouryi is executed concurrently during training, whereas the GPU utilization when running \ouryi alone, as shown in Figure 7(a), is approximately 58\%. Overall, the GPU utilization shown in Fig.~\ref{fig:gputil} is intended to illustrate that as $\alpha$ decreases, the GPU resources used are reduced. However, the displayed values do not represent the actual GPU resource consumption; that is, \ouryi does not actually consume up to 70\% of GPU resources as Fig.~\ref{fig:gputil} shows.
%



\begin{figure}[htbp]
\centering
\includegraphics[width=0.8\linewidth]{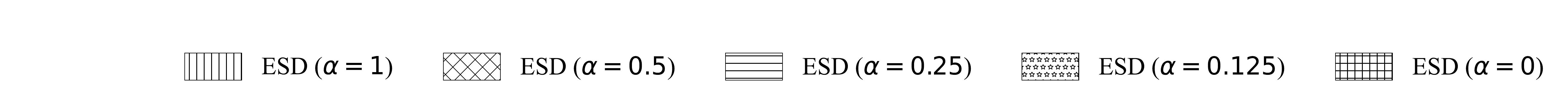}

\vspace{10pt} 

\begin{subfigure}{0.4\linewidth}
\centering
\includegraphics[width=\linewidth]{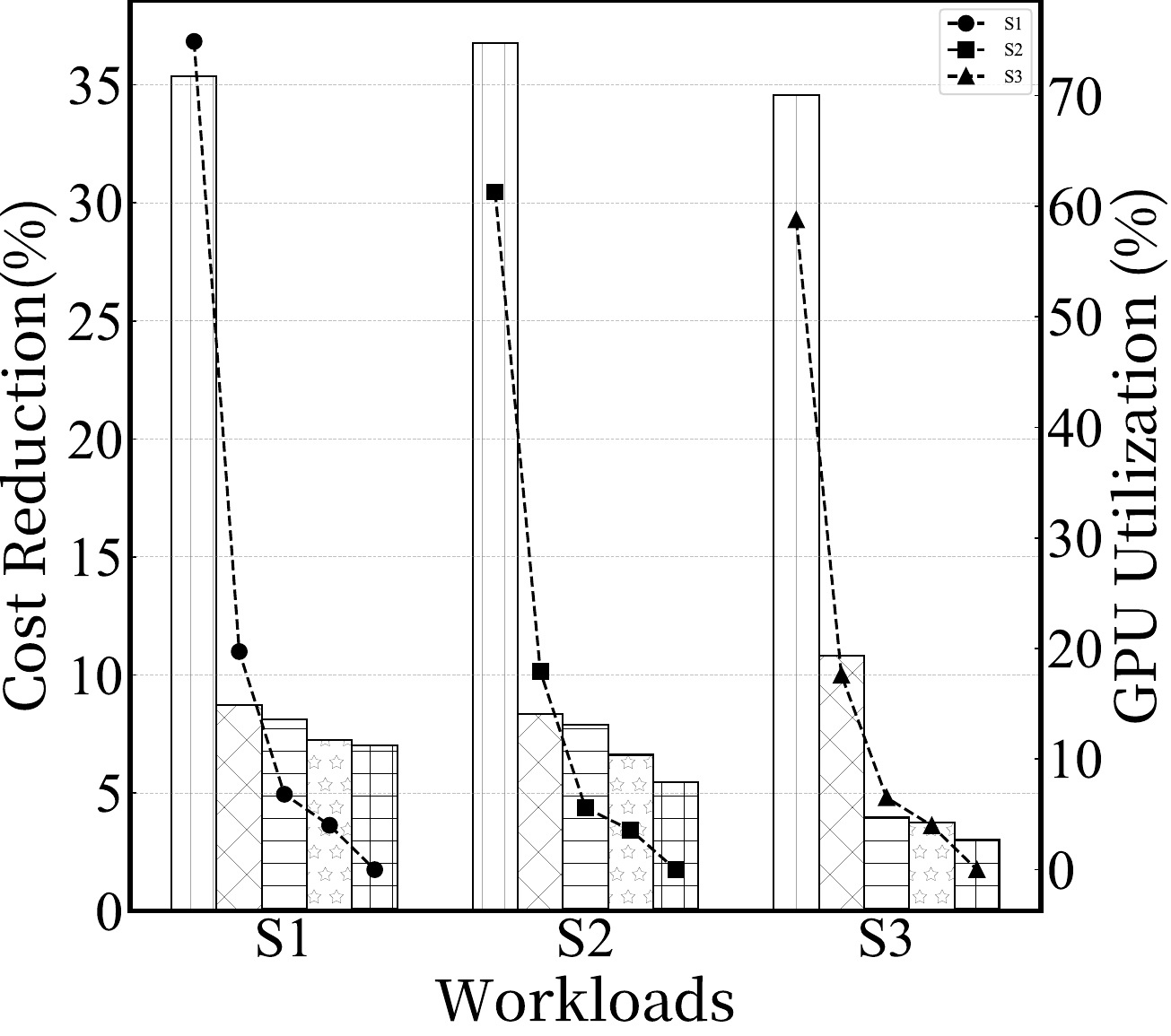}
\caption{Batch size per worker = 128}
\label{fig:gputil128}
\end{subfigure}
\hfill
\begin{subfigure}{0.4\linewidth}
\centering
\includegraphics[width=\linewidth]{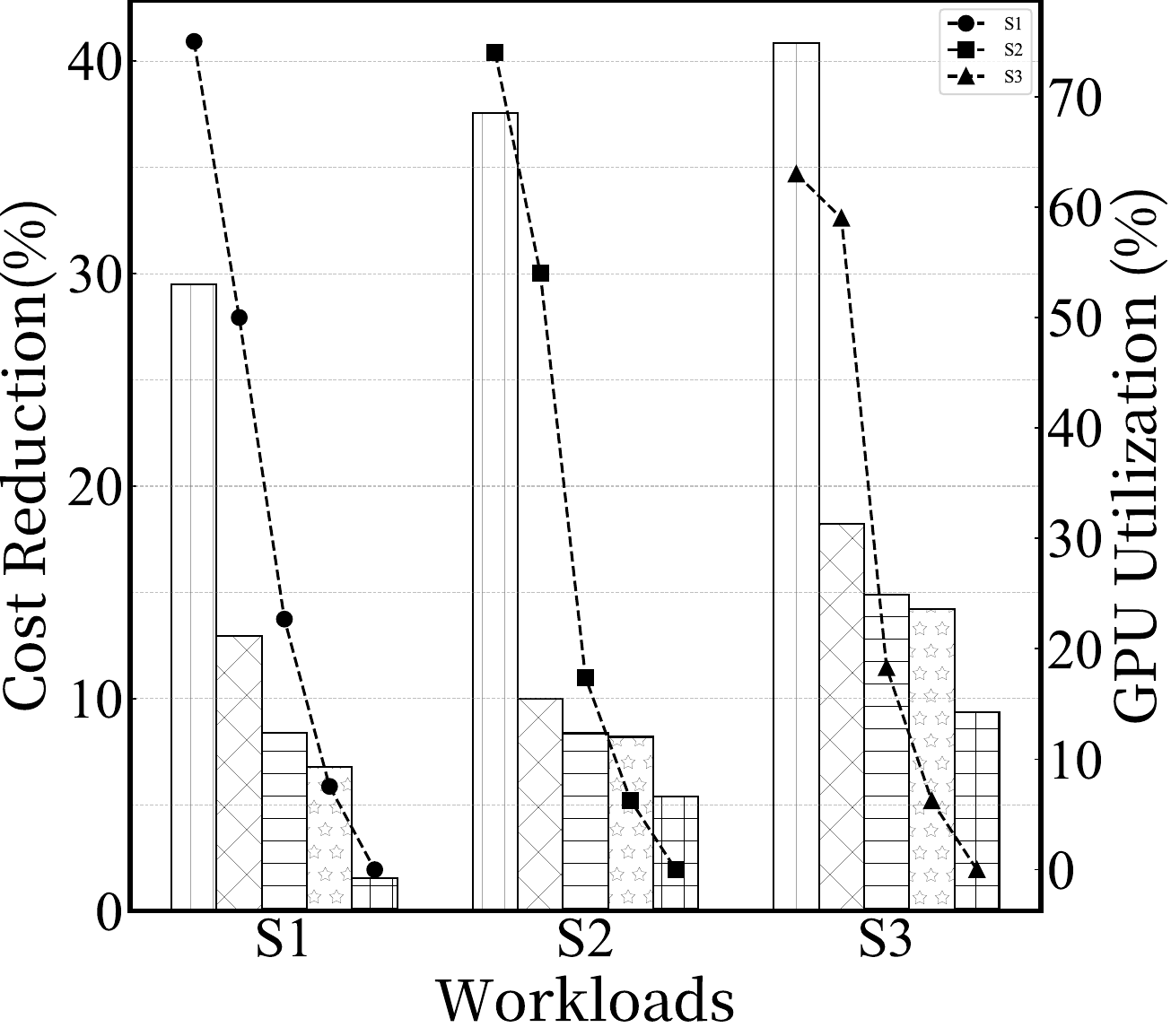}
\caption{Batch size per worker = 256}
\label{fig:gputil256}
\end{subfigure}

\caption{Cost reduction and GPU utilization.} 
\label{fig:gputil}
\end{figure}

\subsection{ Sensitivity Analysis}

Next, we explore the performance of \our when the batch size per worker, cache ratio, and embedding size change. Additionally, we assess the performance of \our with four workers and homogeneous networks (same bandwidth) between workers and the parameter server.  When one setting changes, other settings remain at their default setting, in this part, we mainly focus on workload S2.


\stitle{Batch Size Per Worker.}

As Fig.~\ref{fig:batchsize} shows, to investigate the impact of batch size per worker, we vary it from $64$ to $512$ and present the speedup and the cost reduction. When the batch size per worker increases from $64$ to $256$, the speedup trends for \ouryi, \ourwu, and \ourerwu are increasing, reaching up to $1.54\times$. However, when the batch size per worker increases from $256$ to $512$, the speedup for \ouryi, \ourwu, and \ourerwu does not surpass that at $256$ (although it remains above $1$). This is because increasing the batch size per worker increases the decision-making time for \ouryi and \ourwu, and the solution quality of  \ourling  decreases with the larger batch size per worker. Thus, the speedup for all three mechanisms declines. For the cost reduction of each mechanism, as the batch size per worker increases, the general trends of cost reduction and speedup are similar.  When the batch size per worker increases from $256$ to $512$, although \ouryi achieves greater cost reduction, the increased decision-making time reduces the speedup from $1.54$ to $1.23$.

\begin{figure}[htbp]
\centering
\includegraphics[width=0.8\linewidth]{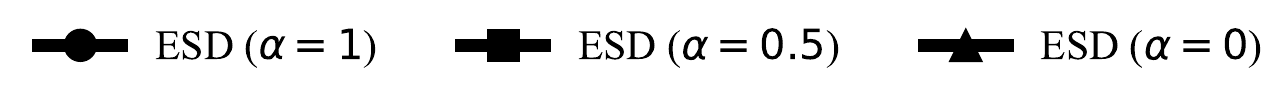}

\vspace{10pt} %

\begin{subfigure}{0.4\linewidth}
\centering
\includegraphics[width=\linewidth]{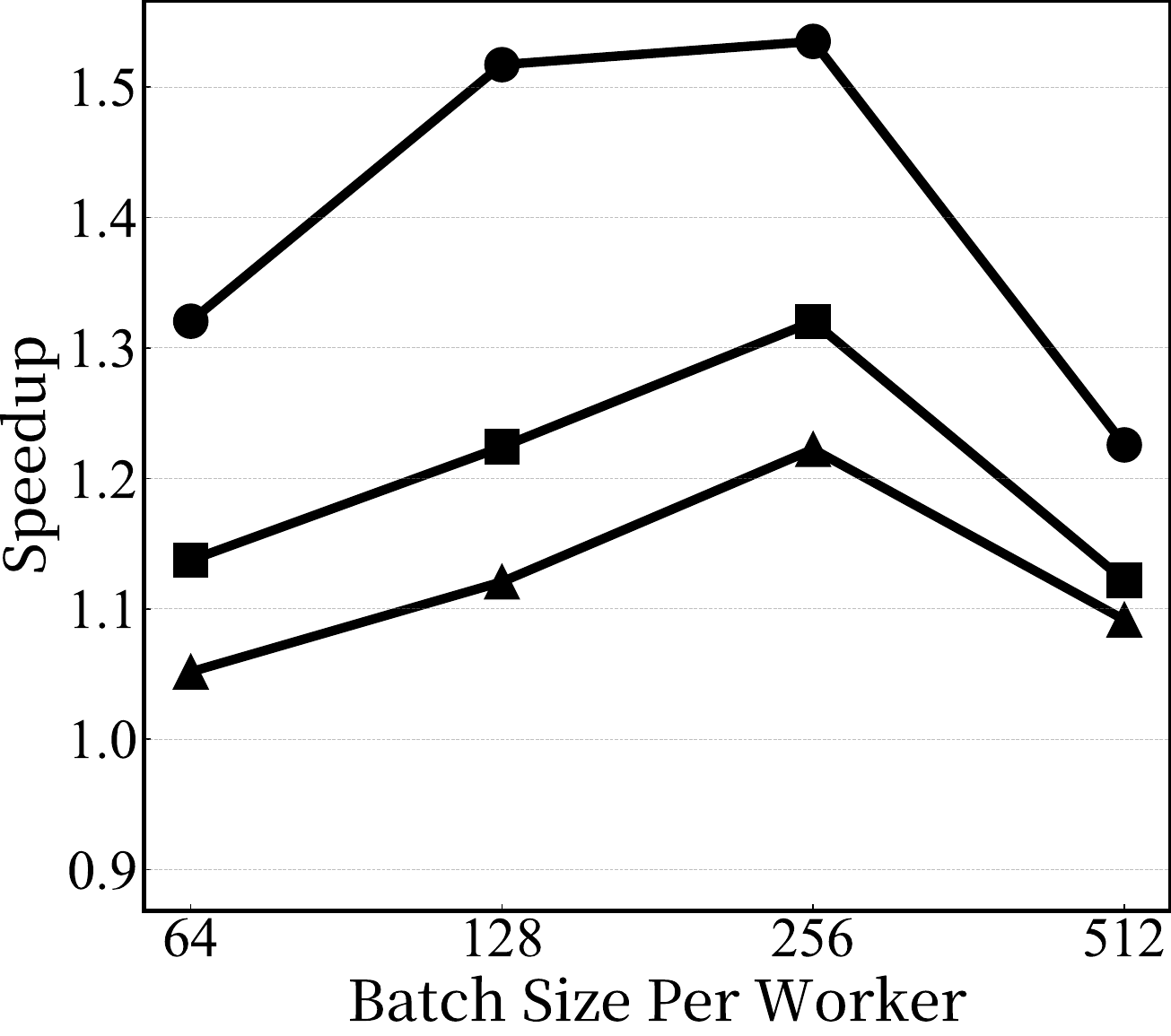}
\caption{ItpS Speedup}
\label{fig:batchsizespeed}
\end{subfigure}
\hfill
\begin{subfigure}{0.4\linewidth}
\centering
\includegraphics[width=\linewidth]{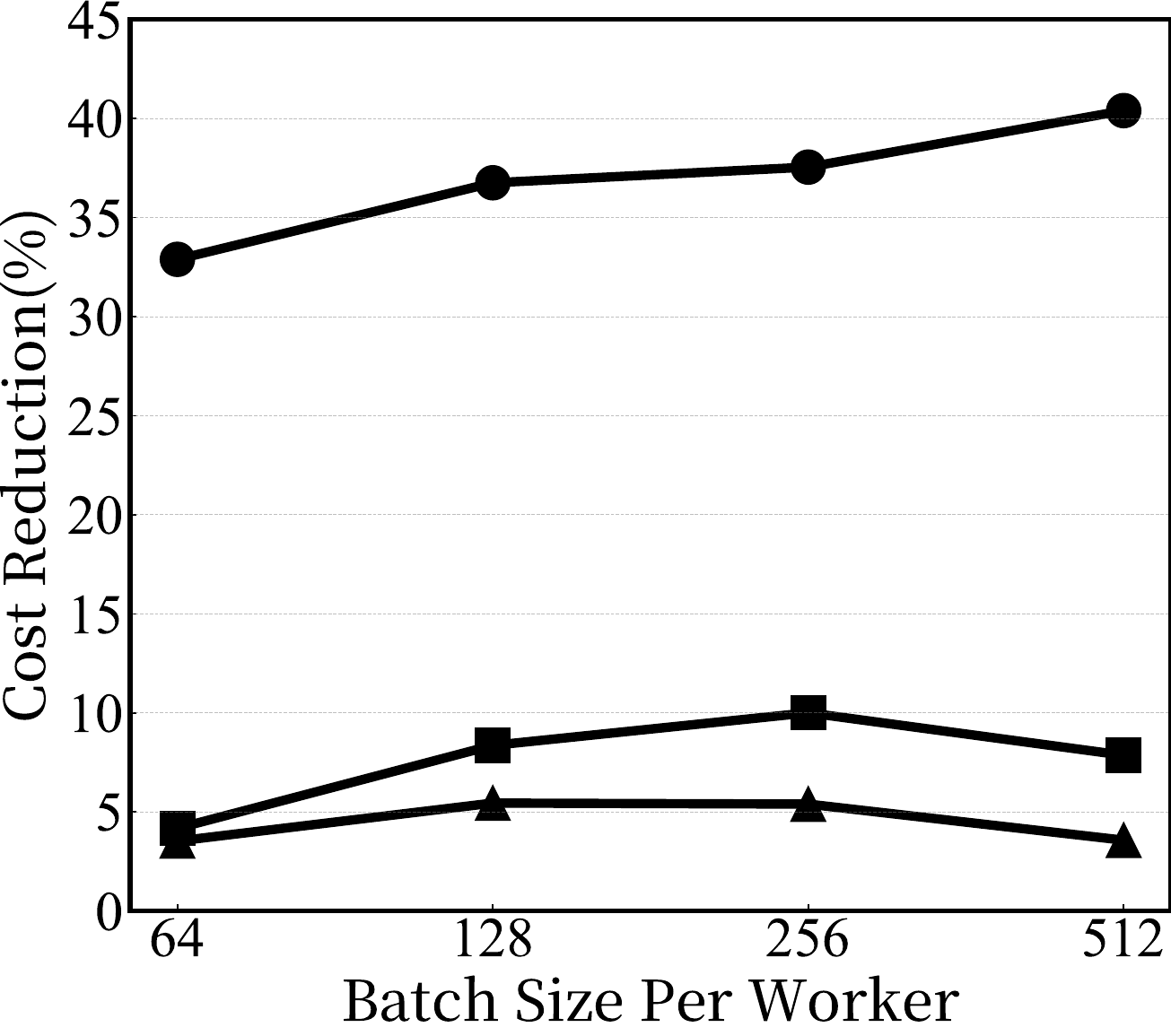}
\caption{Cost Reduction}
\label{fig:batchsizecost}
\end{subfigure}

\caption{Impact of batch size per worker.} 
\label{fig:batchsize}
\end{figure}



\stitle{Cache Ratio.} We show the result of the impact of cache ratio in Fig.~\ref{fig:cachesize}. The cache ratio, \ie, the ratio of the number of in-cache embeddings to the total number of embeddings, ranges from 4\% to 10\%.  Fig.~\ref{fig:cachesize} shows that \our preserves the performance superiority over LAIA and the speedup for the same $\alpha$ does not vary a lot across different sizes of embedding cache. This experiment verifies the consistent effectiveness of \our under different cache ratios.
\begin{figure}[htbp]
\centering
\includegraphics[width=0.8\linewidth]{fig/legend_sca.pdf}


\begin{subfigure}{0.4\linewidth}
\centering
\includegraphics[width=\linewidth]{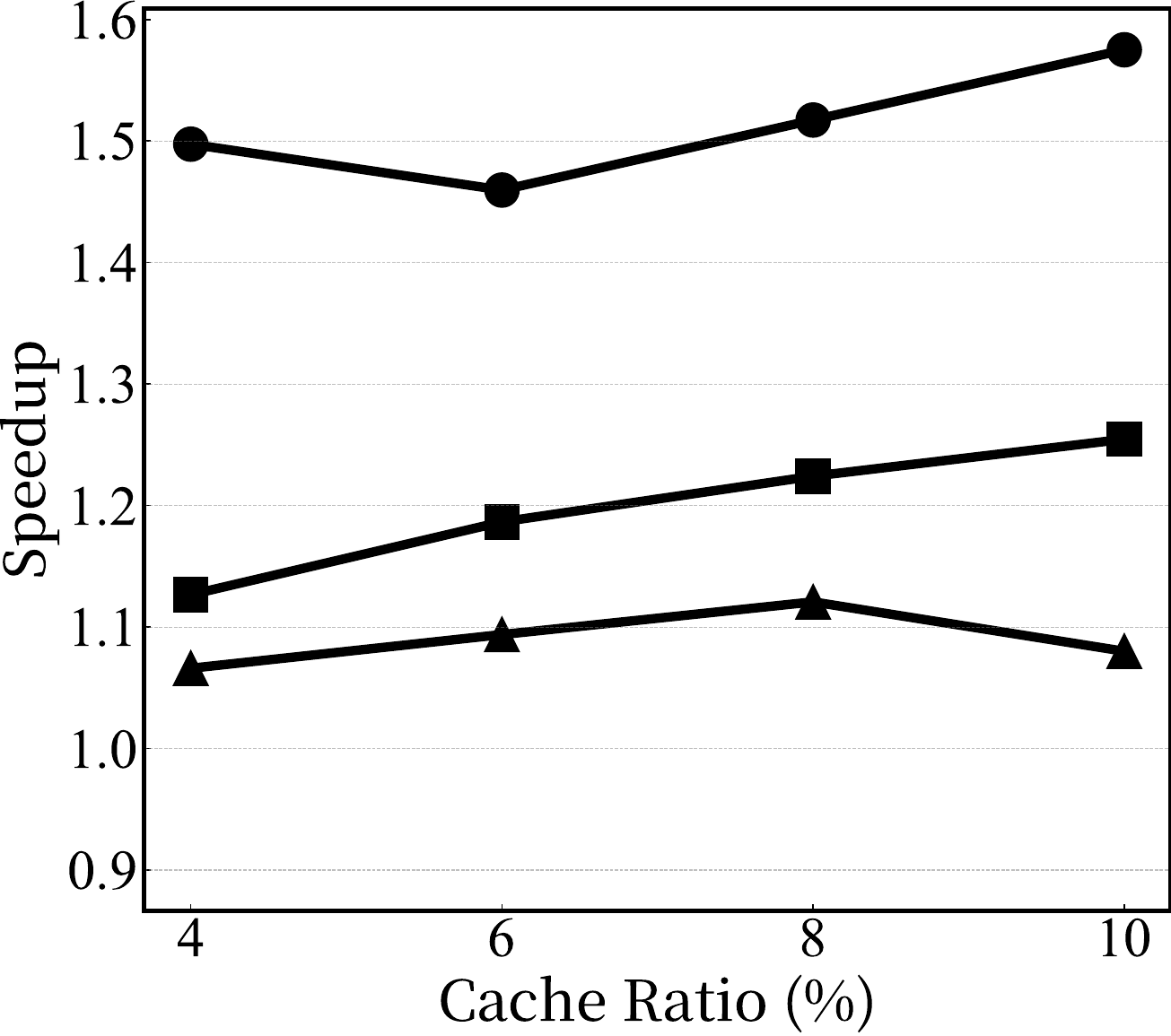}
\caption{ItpS Speedup}
\label{fig:cachesizespeed}
\end{subfigure}
\hfill
\begin{subfigure}{0.4\linewidth}
\centering
\includegraphics[width=\linewidth]{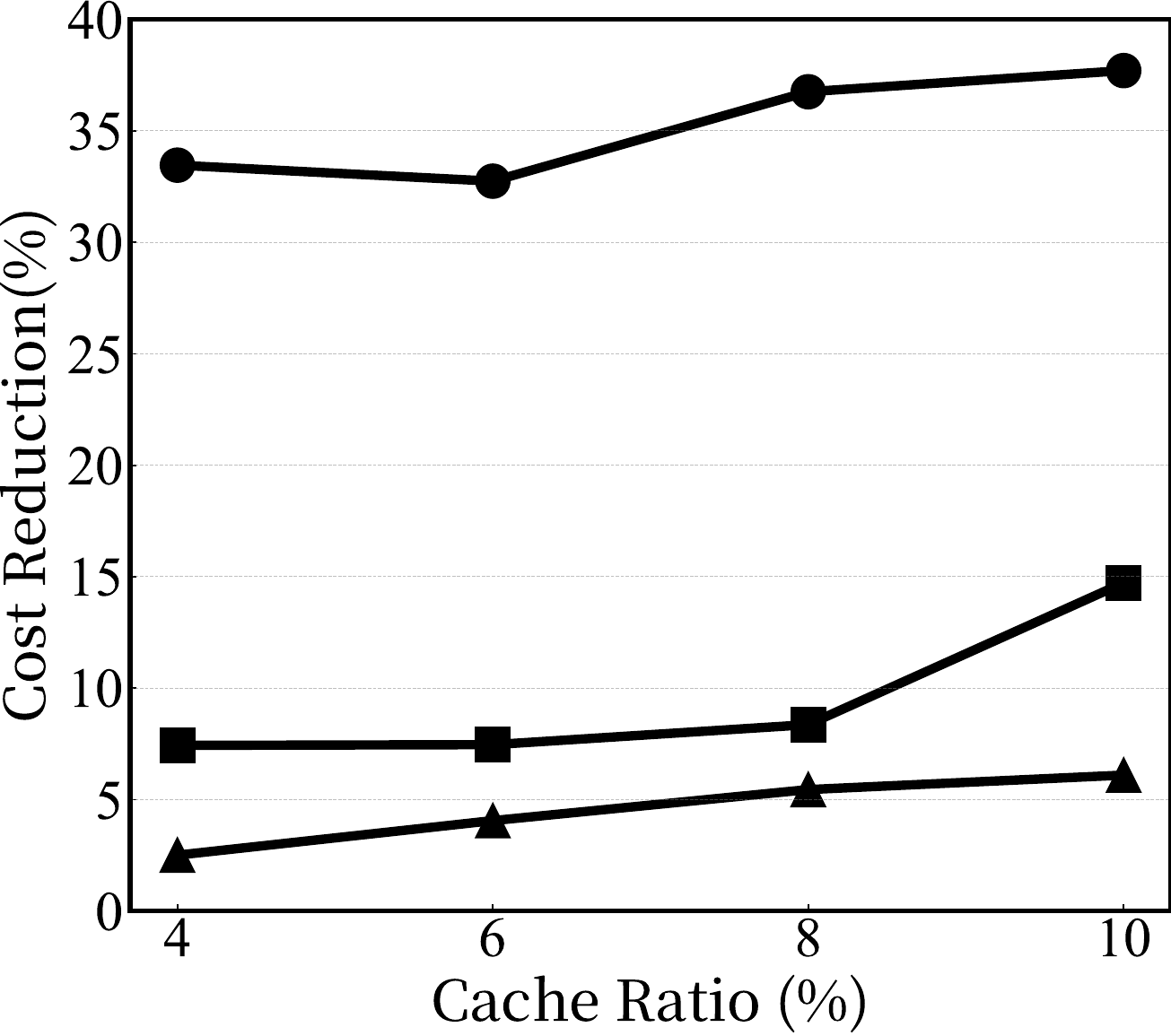}
\caption{Cost Reduction}
\label{fig:cachesizecost}
\end{subfigure}

\caption{Impact of cache ratio.} 
\label{fig:cachesize}
\end{figure}

\stitle{Embedding Size.}
Embedding size \ie, the dimension of the embedding vector, is an important factor affecting the performance of DLRM model. As shown in Fig.~\ref{fig:embeddingsize}, we test the speedup and cost reduction variations of \ouryi, \ourwu, and \ourling for embedding sizes of $128$, $256$, $512$, and $1024$. When the embedding size increases, as shown in Fig.~\ref{fig:embeddingspeed}, the speedup of all three mechanisms compared to LAIA increases, reaching up to $1.59\times$. The reason is that as the embedding size increases, the data size of each embedding pull and push increases, \ie, $D_{tran}$ becomes larger, thereby increasing the embedding transmission cost. When the transmission cost of each embedding increases, the effectiveness of \our becomes more pronounced.  For cost reduction, changes in embedding size only affect \(D_{tran}\) and do not impact the cost reduction of \ouryi, \ourwu, and \ourling relative to LAIA. For instance, when the embedding size increases from 256 to 512, \(D_{tran}\) doubles, causing the embedding transmission cost for both LAIA and \our to double as well. However, the relative cost reduction of LAIA remains unchanged.


\begin{figure}[htbp]
\centering
\includegraphics[width=0.8\linewidth]{fig/legend_sca.pdf}

\vspace{10pt} 

\begin{subfigure}{0.4\linewidth}
\centering
\includegraphics[width=\linewidth]{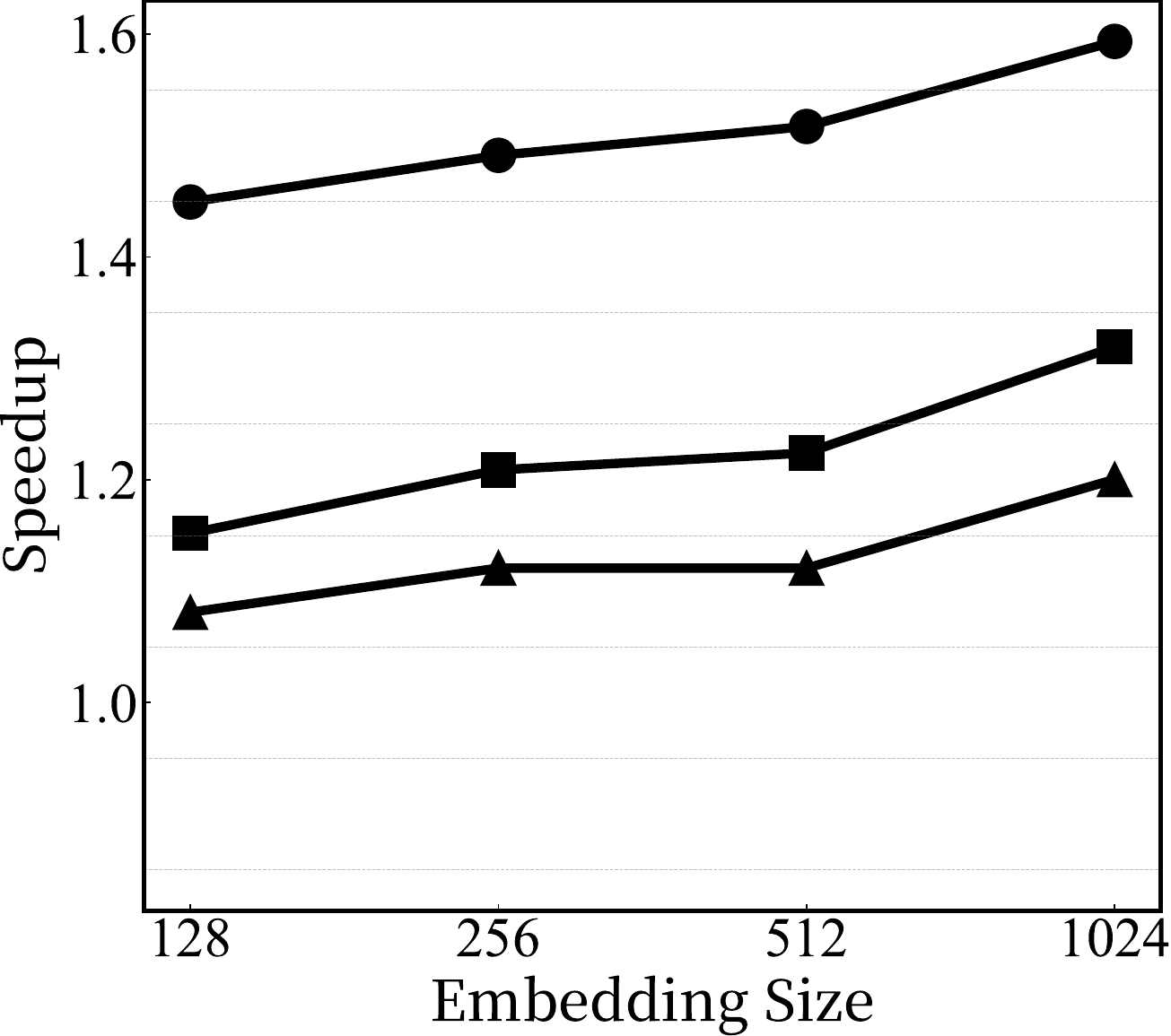}
\caption{ItpS Speedup}
\label{fig:embeddingspeed}
\end{subfigure}
\hfill
\begin{subfigure}{0.4\linewidth}
\centering
\includegraphics[width=\linewidth]{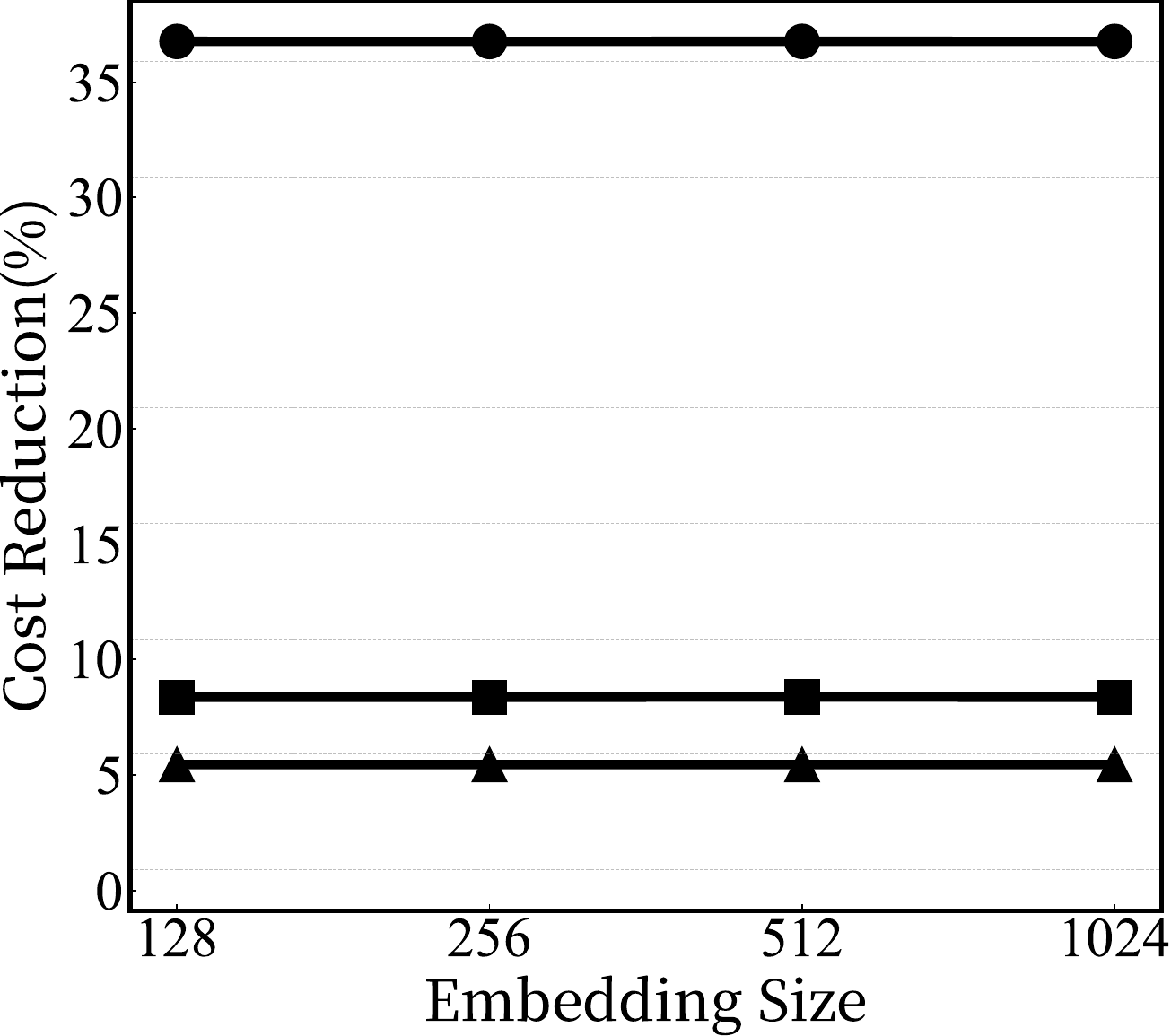}
\caption{Cost Reduction}
\label{fig:embeddingcost}
\end{subfigure}

\caption{Impact of embedding size.} 
\label{fig:embeddingsize}
\end{figure}

\stitle{Worker Number and Network Bandwidths.}

Previous experiments were conducted with $8$ edge workers where the network bandwidth between workers and the PS is heterogeneous. To verify the effectiveness of \our with different numbers of workers and with homogeneous network bandwidth between workers and the PS, we tested the performance of \ouryi, \ourwu, and \ourling under the following settings: 1) Using four edge workers, with two $5$ Gbps workers and two $0.5$ Gbps workers. 2) Using four $5$ Gbps workers. The experimental results are shown in Fig.~\ref{fig:four}. Under two settings, for the three mechanisms and three workloads, the speedup ranges from $1.07\times$ to $1.31\times$ and $1.03\times$ to $1.23\times$, respectively. For cost reduction, when using two 5 Gbps workers and two 0.5 Gbps workers, the cost reduction ranges from $6.06\%$ to $42.15\%$. In contrast, when using four 5 Gbps workers, the cost reduction ranges from $0.33\%$ to $29.11\%$. The experiments show that \our is effective with four workers in homogeneous networks and performs better in heterogeneous networks, aligning with our motivation.


\begin{figure}[h]
\centering
\includegraphics[width=0.8\linewidth]{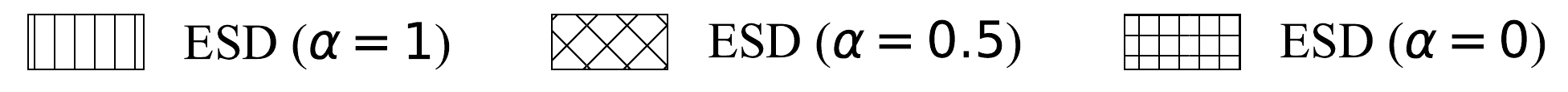}

\vspace{10pt} 

\begin{subfigure}{0.4\linewidth}
\centering
\includegraphics[width=\linewidth]{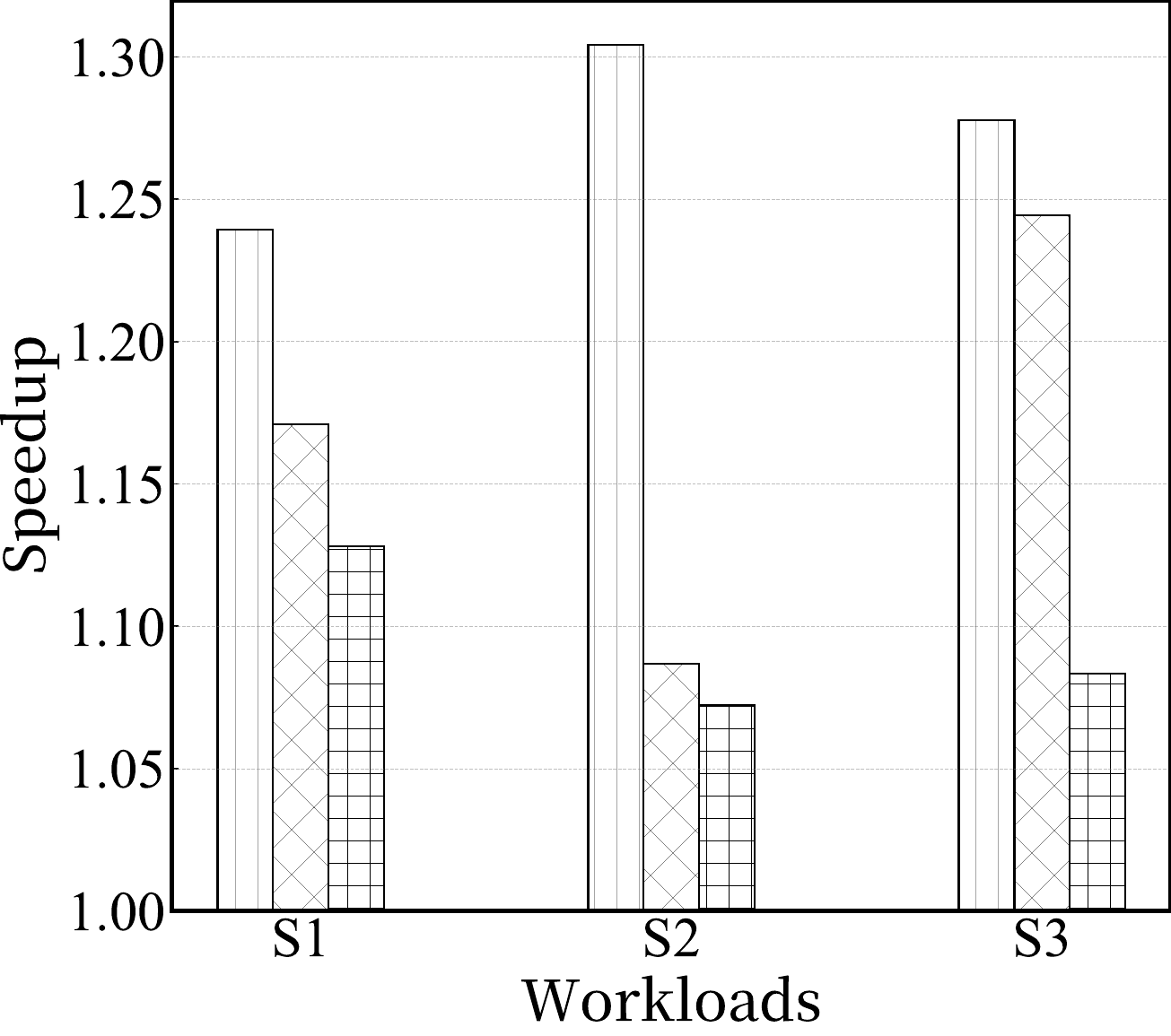}
\caption{ItpS Speedup (Heterogeneous)}
\label{fig:4yigouspeed}
\end{subfigure}
\hfill
\begin{subfigure}{0.4\linewidth}
\centering
\includegraphics[width=\linewidth]{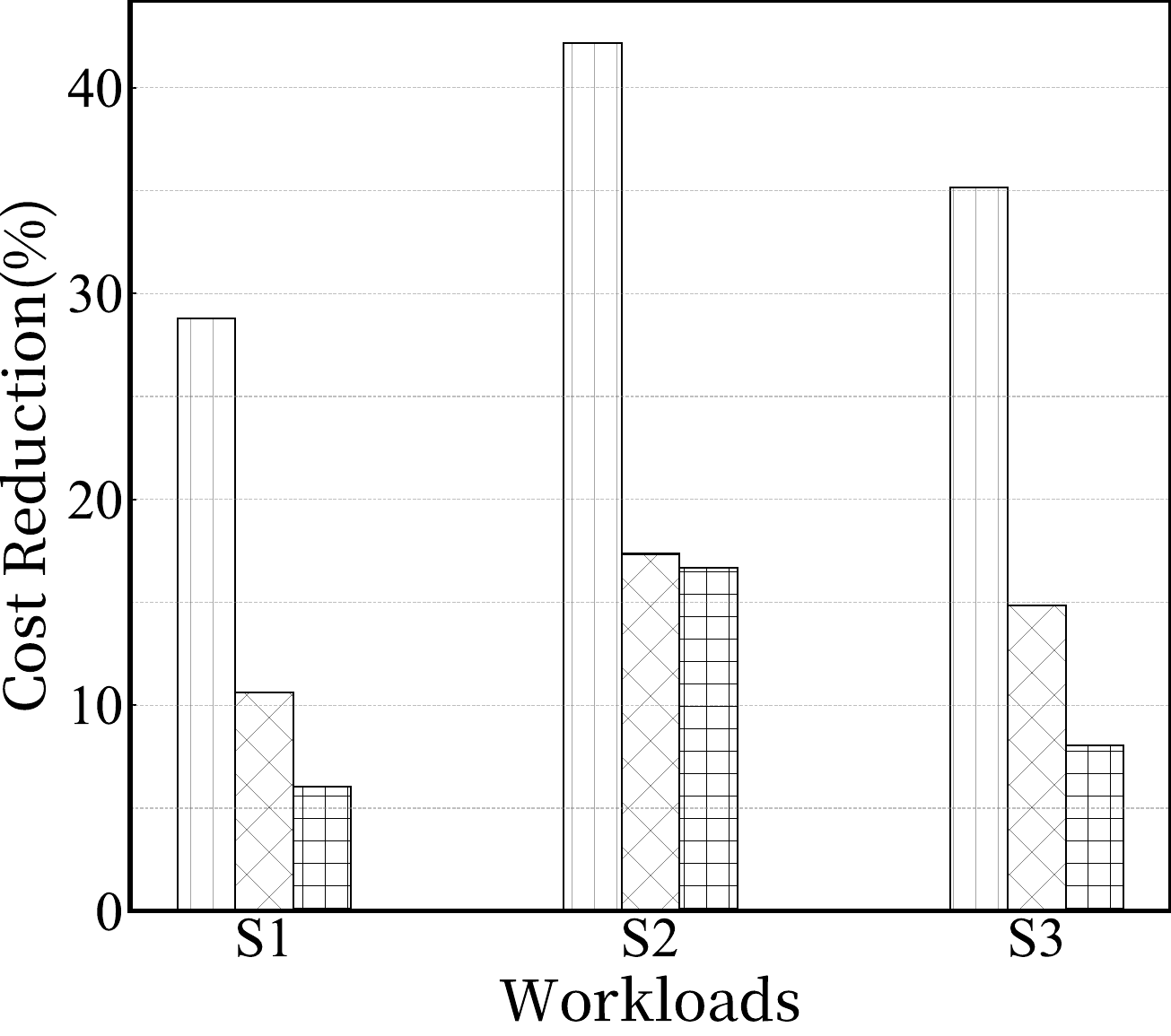}
\caption{Cost Reduction (Heterogeneous)}
\label{fig:4yigoucost}
\end{subfigure}

\vspace{10pt} 

\begin{subfigure}{0.4\linewidth}
\centering
\includegraphics[width=\linewidth]{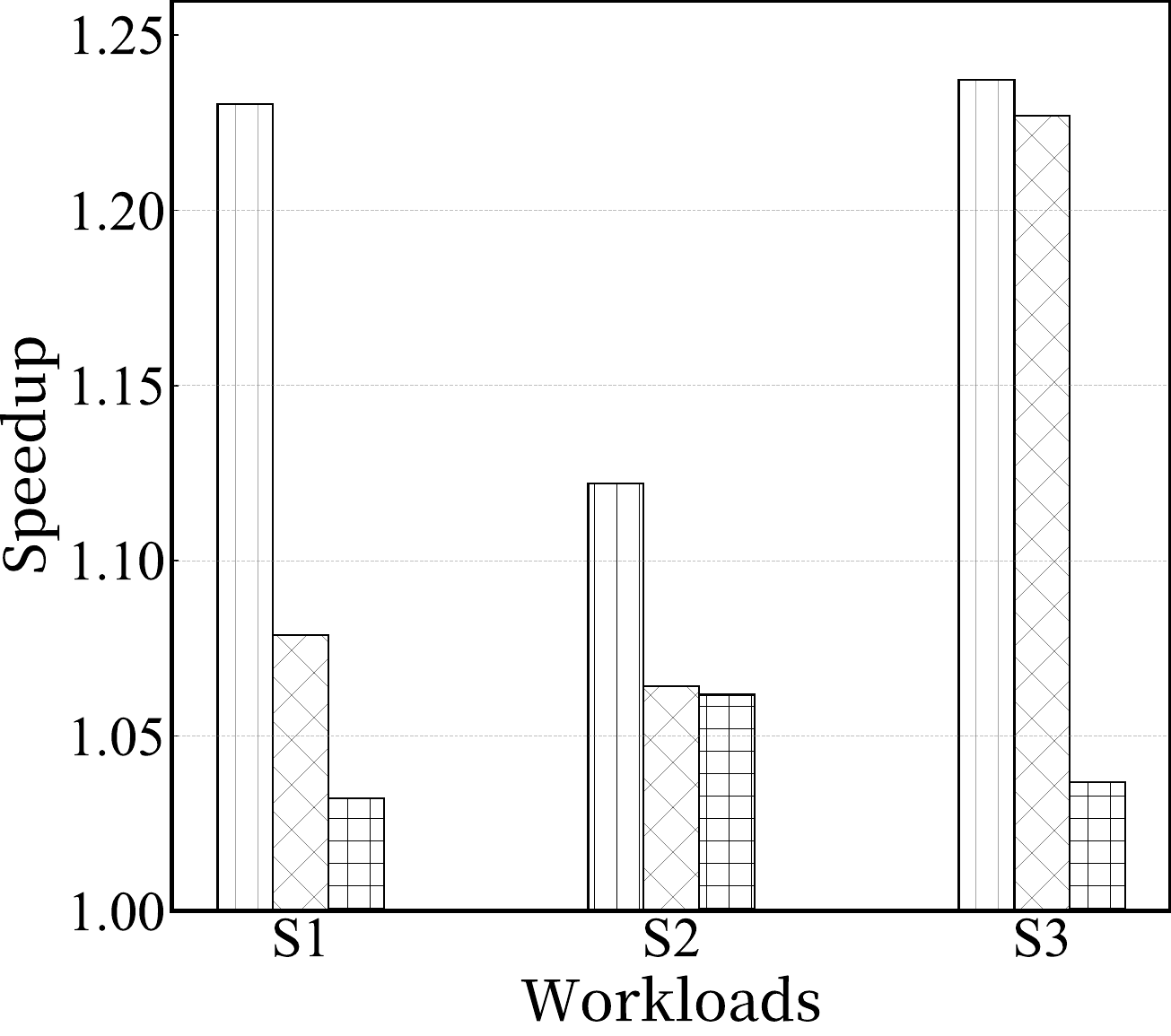}
\caption{ItpS Speedup (Homogeneous)}
\label{fig:4tonggouspeed}
\end{subfigure}
\hfill
\begin{subfigure}{0.4\linewidth}
\centering
\includegraphics[width=\linewidth]{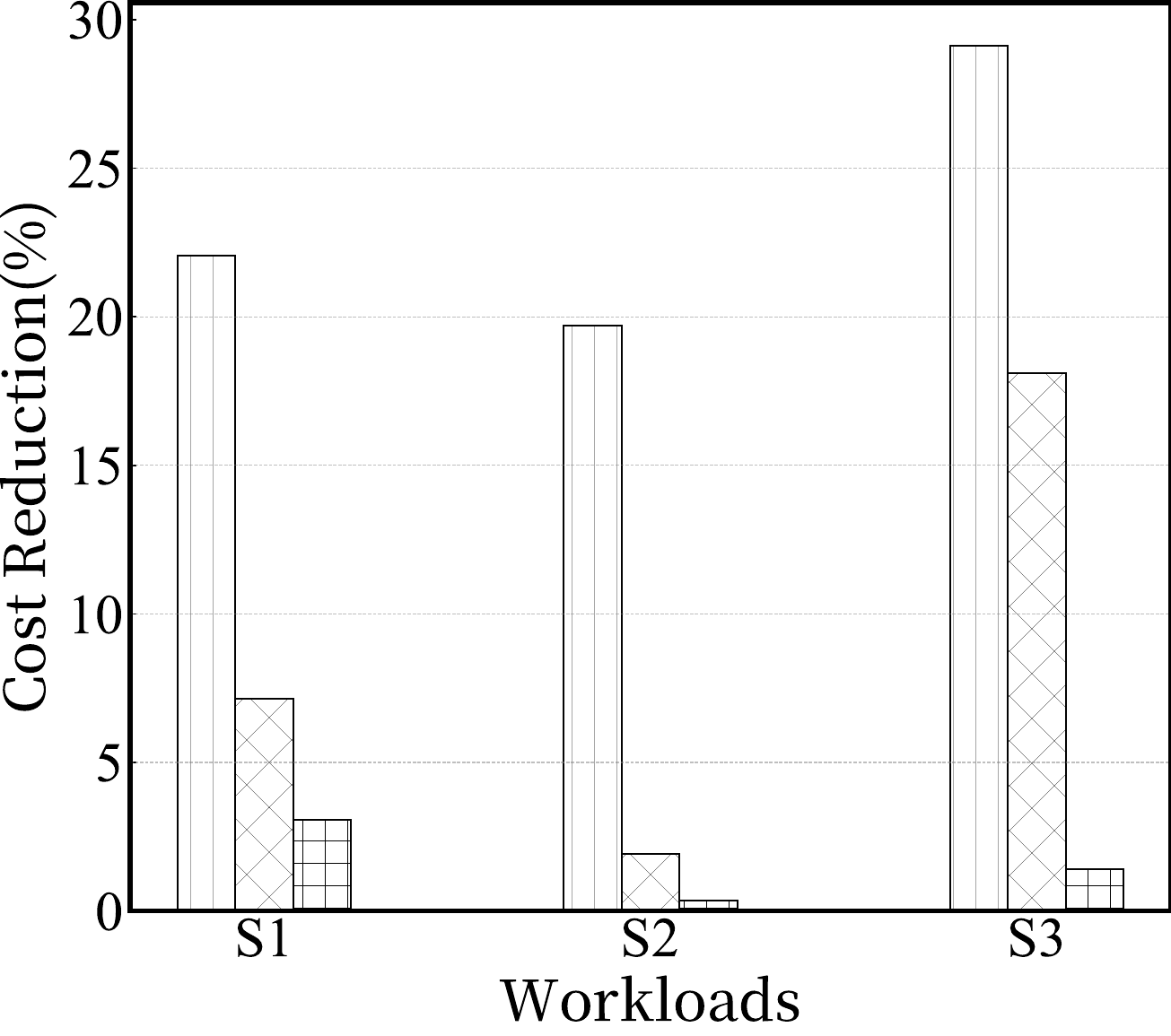}
\caption{Cost Reduction (Homogeneous)}
\label{fig:4tonggoucost}
\end{subfigure}

\caption{Experiment results when using four workers.} 
\label{fig:four}
\end{figure}



\section{Related Works}\label{sec:related}
\subsection{Recommender Systems for Mobile and Ubiquitous Computing}
The recommender systems designed for mobile and ubiquitous computing typically focus on privacy protection and efficiency. These systems are increasingly deployed in environments where user data is distributed across devices, such as mobile phones and edge devices, requiring approaches that reduce communication costs and respect data privacy.

HeteFedRec~\cite{yuan2024hetefedrec} is a federated recommendation framework that allows personalized model sizing for participants, leveraging dual-task learning for additive aggregation, dimensional decorrelation regularization to prevent collapse, and relation-based knowledge distillation to enhance knowledge sharing. Yuan~\etal~\cite{yuan2024hide} propose a novel parameter transmission-free federated recommendation framework, PTF-FedRec, which facilitates collaborative learning via transmitting predictions between clients and the central server, while ensuring privacy through a sampling and swapping mechanism for local model predictions.  RelayRec~\cite{deng2024relayrec} is a privacy-preserving cloud-device relay on-device learning framework for CTR prediction, combining clustered meta-learning to train preference-specific models, an automated model selector for personalized initialization, and collaborative learning to enhance local prediction performance using decentralized private data. HFSA~\cite{li2023hfsa} is a semi-asynchronous hierarchical federated recommendation system that reduces communication overhead through edge aggregation, improves performance with a semi-synchronous mechanism, and ensures global model convergence by allowing slow clients to asynchronously contribute their parameters, making it suitable for dynamic mobile and edge environments. DualRec~\cite{zhang2025dualrec} is a  novel collaborative training framework that combines federated learning with efficient model aggregation and denoising mechanisms, allowing devices to train lightweight models with local data while a larger model is trained on the cloud, enhancing recommendation system performance by generating pseudo user interactions and distilling knowledge between the device and cloud models. PREFER~\cite{guo2021prefer} is an edge-accelerated federated learning framework for Point-of-Interest recommendation, where user-independent parameters are shared among users and aggregated on edge servers. PEPPER~\cite{belal2022pepper} is a decentralized recommender system based on gossip learning principles, enabling users to train personalized models asynchronously. This paper focuses on providing recommendation services for mobile and ubiquitous applications, emphasizing the improvement of training efficiency for deep learning recommendation models with huge embedding tables in edge environments. It accelerates model training by reducing embedding transmission costs through embedding samples dispatching, making it  suitable for resource-limited devices in mobile  and  ubiquitous computing scenarios.

\subsection{Accelerate Deep Learning Recommendation Models Training}

Part of prior works focuses on minimizing the number of necessary transmissions~\cite{adnan2021accelerating,lian2022persia,miao2021het}. FAE~\cite{adnan2021accelerating}  decreases embedding transfers from CPUs to GPUs by storing popular embeddings on GPUs. Persia~\cite{lian2022persia} advocates mixing the synchronous and asynchronous mechanisms to update MLP and embedding tables, respectively. HET~\cite{miao2021het} develops a caching consistency model that allows staleness in cache read and write operations to minimize transmission overhead, while LAIA~\cite{zeng2024accelerating} schedules embeddings to reduce transmission cost without sacrificing accuracy. Other efforts focus on reducing   per-transmission latency. Ugache~\cite{song2023ugache} introduces a decomposition extraction mechanism to prevent bandwidth congestion and utilizes NVLink and NVSwitch to accelerate remote GPU memory access.   AdaEmbed~\cite{lai2023adaembed} and CAFE~\cite{zhang2024cafe} aim to reduce the transmission cost of embeddings by compressing and pruning embeddings. ScaleFreeCTR~\cite{guo2021scalefreectr} employs a mixcache mechanism to eliminate transmission cost between hosts and GPUs, and leverages virtual sparse ID  to reduce transmission volume.  In this paper, we focus on training DLRM in edge environments. We propose \our to address the unique challenges of edge environments, such as limited resources and heterogeneous networks, to minimize the total embedding transmission cost and thus accelerate DLRM training.

\section{Discussion}\label{sec:dis}

\subsection{Cache Replacement Policy}

During the training of DLRM, when the embedding cache reaches capacity, one or more cached embeddings must be evicted to accommodate new ones. Under on-demand synchronization, if the gradients of the evicted embeddings have not yet been synchronized with the PS, an evict push operation is triggered.   In this work, we design  \ourlru, a marking-based   policy to manage cached embeddings  in each worker to reduce the number of evict push operations, considering recency, frequency, and the version of embedding.  When embedding id $x_i$ is dispatched to $w_i$, \ourlru assigns  a mark (\target) to $Emb(x_i)$, \target is a positive integer starting at $1$. When the cache is full, and all embeddings'   marks are \target, \target$+=1$. For cache replacement, when the cache is full, \ourlru first evicts outdated embeddings, then compares marks in ascending order, and finally compares access frequency. In our implementation, we overload the $<$ operator in C++ to achieve this, with the latest version marked as $1$ and the outdated version as $0$. Additionally, by adjusting the increment method for \target and the comparison order in the overloaded $<$ operator, \ourlru can easily adapt to different priority settings among frequency, recency, and version.  Since Fig.~\ref{fig:ingrecost} shows  that evict push contributes less than $10\%$ of  transmission operation, we discuss the cache replacement policy here rather than in Sec.~\ref{sec:design}. Additionally, incorporating the embedding eviction and  replacement policy into calculating the expected transmission cost is an effective way to reduce evict push operations, and we leave it for future work.

\subsection{Non-uniform Embedding Size}

In this paper, we propose \our to dispatch input  embedding samples to workers based on calculating the expected transmission cost under the assumption of uniform embedding vector dimensions (\ie, embedding size). However, recent trends indicate the use of non-uniform embedding sizes to reduce the memory footprint, training time, and inference time of DLRM~\cite{zhaok2021autoemb, luo2024fine, lai2023adaembed}. When embedding sizes are non-uniform, \our can effectively handle this through the following adjustments. First, for calculating the expected transmission cost, different data sizes ($D_{tran}$) can be used to account for varying embedding dimensions. Secondly, regarding the cache replacement policy, non-uniform embedding sizes result in differing memory footprints within the embedding cache. This can be managed using a priority metric. The priority for each embedding in the cache is calculated by the frequency, recency, and version (latest or outdated) as the numerator, while the memory footprint of each embedding serves as the denominator. When the cache reaches capacity, embeddings with the lowest priorities are evicted first. A future direction is to explore embedding size configurations for training and deploying DLRM in edge environments.

\section{Conclusion}\label{sec:conclu}
Training and deploying DLRM on edge workers is an effective way to deliver recommendation services for modern mobile and  ubiquitous applications while ensuring privacy protection and low latency. This paper focuses on dispatching input embedding samples based on expected embedding transmission costs to minimize the total embedding transmission cost. We highlight the challenges of designing a dispatch mechanism in edge environments, including cost composition, heterogeneous networks, and limited resources. To address these challenges, we propose a dispatch mechanism called \our, incorporating a hybrid decision-making method \ourmix as a key component. We implement \our using C++ (including CUDA) and Python and conduct experiments on  edge workers with typical workloads to validate the improvement over the state-of-the-art mechanism. Building an efficient system for modern mobile and  ubiquitous applications requires optimization in data collection, processing, model training, deploying, and more. We hope this work contributes a step towards building an  efficient recommender system for mobile and  ubiquitous computing. This work does not raise any ethical issues.

\clearpage
\bibliographystyle{plain}
\bibliography{references}

\end{document}